\documentclass[reqno]{amsart}

\usepackage{a4,amsmath,amssymb,amsthm}
\usepackage{xfrac}

\usepackage{amsaddr}

\usepackage[T1]{fontenc}
\usepackage[section]{placeins}

\usepackage{tikz}
\usetikzlibrary{shapes,decorations}
\usetikzlibrary{arrows}

\newcommand{\C}[1]{\textcolor{red}{#1}} 

\renewcommand{\C}[1]{#1}

\makeatletter
  \def\moverlay{\mathpalette\mov@rlay}
  \def\mov@rlay#1#2{\leavevmode\vtop{%
     \baselineskip\z@skip \lineskiplimit-\maxdimen
     \ialign{\hfil$#1##$\hfil\cr#2\crcr}}}
\makeatother

\usepackage{stmaryrd}

\usepackage{url}
\usepackage{hyperref}

\usepackage{color}

\numberwithin{equation}{section}

\usepackage[mathscr]{eucal}

\usepackage{graphicx}
\graphicspath{pic/}
\usepackage{psfrag}
\usepackage{blindtext}
\usepackage[inline]{enumitem}
\usepackage[labelfont=bf]{subfig}

\newcommand{\Reals}{{\mathbb R}}         

\newcommand{\HH}{\mathcal H}

\DeclareFontFamily{OT1}{rsfs}{}
\DeclareFontShape{OT1}{rsfs}{m}{n}{ <-7> rsfs5 <7-10> rsfs7 <10-> rsfs10}{}
\DeclareMathAlphabet{\mycal}{OT1}{rsfs}{m}{n}

\newcommand{\scri}{{\mycal I}}

\def\zb{\overline{z}}
\def\Moeb{\mbox{Mb}}
\def\Id{\mbox{Id}}
\def\C{\overline{\mathbb{C}}}

\theoremstyle{plain}
\newtheorem{thm}{Theorem}

\newtheorem{lemma}[thm]{Lemma}

\newtheorem{definition}[thm]{Definition}

\newtheorem{remark}[thm]{Remark}

\makeatletter
\@addtoreset{case}{thm}
\makeatother

\newcounter{mnotecount}[section]

\newcounter{mymnotecount}[section]

\title{The Fingerprints of Black Holes - Shadows and their Degeneracies}
\author[M. Mars, C.~F. Paganini, and M.~A. Oancea]{Marc Mars$^\dagger$, Claudio F. Paganini$^\ddagger$, and Marius A. Oancea$^\ddagger$} 
\email{marc@usal.es}
\email{claudio.paganini@aei.mpg.de}
\email{marius.oancea@aei.mpg.de}
\address{$^\dagger$ 
Instituto de F\ifmmode \acute{i}\else \'{i}\fi{}sica Fundamental y Matem\ifmmode \acute{a}\else \'{a}\fi{}ticas, Universidad de
Salamanca, Plaza de la Merced s/n, 37008 Salamanca, Spain}
\address{$^\ddagger$Albert Einstein Institute, Am M\"uhlenberg 1, D-14476 Potsdam,  Germany }

\begin{document}

\date{\today \ {\em File:\jobname{.tex}}}

\begin{abstract} We show that, away from the axis of symmetry, no continuous degeneration exists between the shadows of observers at any point in the exterior region of any Kerr-Newman black hole spacetime of unit mass. Therefore, except
  possibly for discrete changes, an observer can, by measuring the black holes shadow, determine the angular momentum and the charge of the black hole under observation, as well as the observer's radial position
  and angle of elevation above the equatorial plane.
  Furthermore, his/her relative velocity compared to a standard observer can also be measured.
  On the other hand, the black hole shadow does not allow for
  a full parameter resolution in the case of a Kerr-Newman-Taub-NUT black hole, as 
  a continuous degeneration relating specific angular momentum,
  electric charge, NUT charge and elevation angle exists in this case. 
  \end{abstract}

\maketitle

\tableofcontents

\section{Introduction}
The shadow of the black hole is defined as the set of trajectories on which no light from a background source, passing a black hole, can reach the observer. There is hope for the Event Horizon Telescope to be able to resolve the black hole at the center of the Milky Way (Sgr A*) well enough that one can compare it to the predictions from theoretical calculations, see for example \cite{doeleman_event-horizon-scale_2008}. Therefore, analyzing the shadows of black holes is of direct astronomical interest. The first discussion of the shadow in Schwarzschild spacetimes can be found in \cite{synge_escape_1966}, and, for extremal Kerr at infinity, it was later calculated in \cite{bardeen_black_1973}. The perspective of an actual observation has led to a number of advancements in the theoretical treatment of black hole shadows in recent years \cite{1973blho.conf..215B,scalarhair,grenzebach_photon_2014,grenzebach_photon_2015,SCHEE2009,Takahashi2010,Teo2003}. 
In \cite{hioki_measurement_2009,li_measuring_2014} possible ways to extract the black hole parameters from the observation of the shadow have been explored. However, a strict treatment of the question "How much information about the black hole is there in the shape of the shadow?" has, to our knowledge, not been carried out. The only work we are aware of, that takes into account the fact that an observer cannot a priori know what his detailed motion with respect to the manifold is, can be found in \cite{grenzebach_aberrational_2015}. However, the focus in that work is more on the explicit deformation due to different velocities rather than a systematic study on how the freedom of picking any observer at a point influences the possibility of extracting information about the black hole from the shape of the shadow.\\
The most important conceptual idea introduced in the present work is the notion of what it means for the shadows at two points to be degenerate. In the case of degeneracy there exist two distinct observers for which the shadow is absolutely identical. Consequently, an observer cannot distinguish - from shape and size of the shadow alone - between the two situations.\\
We will put the concept of degeneracy to work in this paper by proving the existence of two continuous degeneracies, one parameter curves in the parameter space of observers in the exterior region of Kerr-Newman-Taub-NUT spacetimes. Beyond that, we show that there are no further continuous degeneracies. 

Note that even though the shadow parametrization in \cite{grenzebach_photon_2015} is given for the entire Pleba\ifmmode \acute{n}\else \'{n}\fi{}ski-Demia\ifmmode \acute{n}\else \'{n}\fi{}ski class of black hole spacetimes, we will focus on the subclass of Kerr-Newman-Taub-NUT black holes in the present work. This has two reasons, first for observations within our galaxy the cosmological constant should be negligible, second including the cosmological constant increases the complexity of the arguments without providing additional insight. The case with cosmological constant will be treated in a separate paper.

\subsection*{Overview of this paper}
In section \ref{sec:kerr} we collect some background on the Kerr-Newman-Taub-NUT spacetime. We discuss its properties and discuss the geodesic equation  in these spacetimes in section \ref{sec:geodeq}. 
In section \ref{sec:sphere} we introduce the framework for the discussion of the black hole shadows, in particular the notion of the celestial sphere.  In section \ref{sec:ssshadow} we discuss the shadow for observers at points of symmetry. We use this context to introduce the formal definition of degeneracies and how they arise.
In section \ref{sec:shadowparam} we recall the explicit form of the shadow in the spacetimes we consider.
In section \ref{sec:degeneracies} we introduce the recipe of the search for continuous degeneracies. Finally, in section \ref{sec:main} we present the proof for the main result of our paper. Appendix  \ref{app:A} 
is devoted to deriving
several results on M\"obius transformations needed in the main text
and a list of  somewhat long, explicit expressions have been shifted to  Appendix \ref{app:B}.

\section{The Kerr-Newman-Taub-NUT Spacetime} \label{sec:kerr}
In the following introduction of the spacetimes discussed in this paper, we follow the work of Grenzebach et al. \cite{grenzebach_photon_2014}. However, we are setting the cosmological constant to zero. 
The Kerr-Newman-Taub-NUT black holes are stationary, axially symmetric, type D spacetimes. In Boyer-Lindquist coordinates, $(t,r,\theta,\phi)$, the metric is given by \cite[p. 314]{griffiths2009}:
\begin{equation}
\label{eq:metric}
\begin{split}
    \mathrm{d}s^2 =& \Sigma \left(\frac{1}{\Delta}\mathrm{d} r^2 + \mathrm{d}\theta^2 \right)+ \frac{1}{\Sigma}\left((\Sigma+ a \chi)^2 \sin^2\theta -\Delta \chi^2\right)\mathrm{d} \phi^2 + \\
    & \frac{2}{\Sigma}\left(\Delta\chi-a(\Sigma+a\chi)\sin^2 \theta \right)  \mathrm{d}t\mathrm{d}\phi- \frac{1}{\Sigma}\left(\Delta-a^2\sin^2 \theta \right)\mathrm{d}t^2,
 \end{split}
\end{equation}
   where
\begin{align*}
\Sigma &= r^2 + (l+a \cos  \theta)^2, \\
\chi & =a\sin^2 \theta-2 l (\cos \theta +C), \\
\Delta & = r^2 - 2Mr + a^2-l^2+Q^2.
\end{align*}
The coordinate $t$ takes values in $(-\infty,\infty)$, $r$ in
$(r_+,\infty)$, where $r_+$ is the largest root of $\Delta=0$, while $\theta$ and $\phi$ are standard coordinates on the two-sphere. The relevant physical parameters in the metric are the mass $M$ and charge $Q$, the former will be assumed to be positive
based on physical reasons. Further, we will assume w.l.o.g. that the spin parameter satisfies $a\geq0$. The NUT parameter $l$ which can be interpreted as a gravitomagnetic charge can in principle take any value in $\Reals$.\\
In addition, there is a parameter $C$, first introduced by Manko and Ruiz \cite{Manko-Ruiz}, that induces pathologies on both parts of the rotation axis, unless $C=\pm1$. In the present work we will simply ignore these pathologies and take $C=\pm1$ and only consider the regular part of the rotation axis in these cases (the case $C=-1$ corresponds to the original definition of the NUT metric \cite{NUT1963,Manko-Ruiz}). For a detailed discussion of the rotation axis in the case $l\neq0$ and $a\neq0$ see \cite{TaubNUTsingularities}.\\
We will frequently make use of the following orthonormal tetrad at point $p$:
\begin{align}\label{eq:tetrad}
e_0 &= \left.\frac{(\Sigma + a \chi)\partial_t+ a\partial_\phi}{\sqrt {\Sigma\Delta}}\right|_p ,&\qquad e_1&=\left.\sqrt{\frac{1}{\Sigma}}\partial_\theta\right|_p, \\ \nonumber
e_2&=\left.\frac{-(\partial_\phi+\chi \partial_t)}{\sqrt{\Sigma}\sin \theta}\right|_p,&\qquad e_3&=\left.-\sqrt{\frac{\Delta}{\Sigma}} \partial_r\right|_p. 
\end{align}
The metric does not depend on the coordinates $t$ and $\phi$ and therefore features at least two independent Killing vector fields independent of the choice of the parameters.\\
The Kerr-Newman-Taub-NUT family of metrics contains the Schwarzschild ($a=Q=l=0$),  Kerr ($Q=l=0$), Reissner-Nordstr\"om ($a=l=0$), Kerr-Newman ($l=0$), and Taub-NUT ($a=Q=0$) metrics as special cases.\\
Note that we are only interested in the exterior region of the black hole, hence the region 
at $r>r_+$, where the two horizons $r_{\pm}$ are defined by:
\begin{equation}
r_\pm=M\pm\sqrt{M^2-a^2+l^2- Q^2},
\end{equation}
and we are assuming that the inequality $M^2-a^2+l^2+ Q^2\geq0$ holds.
Note that $\Delta>0$ is satisfied in the exterior region.\\
\subsection{Geodesic Equation}\label{sec:geodeq}
We now focus our attention on null geodesics. For all members of the Kerr-Newman-Taub-NUT family of spacetimes there exist four linearly independent constants of motions for the geodesic equation. The norm of the tangent vector is 
directly related to the mass of the test body:
\begin{equation}
m=g_{\mu\nu}\dot\gamma^\mu\dot\gamma^\nu
\end{equation}
which we will assume to be equal to zero from here on. 
The dot denotes differentiation with respect to the affine parameter
$\lambda$.
The two quantities arising from the Killing vector fields $\partial_t$, $\partial_\phi$:
\begin{equation}
E=-(\partial_t)_\mu \dot\gamma^\mu, \qquad L_z=(\partial_\phi)_\mu\dot\gamma^\mu
\end{equation}
are, for appropriate choice of the parameter along the geodesic, 
the test body's energy and angular momentum  in the direction of 
the axis of symmetry. The fourth constant of motion is called Carter's constant, $K$, and it originates from the existence of a Killing tensor, given by:
\begin{equation}
\sigma_{\mu \nu} = \Sigma ( (e_1)_\mu (e_1)_\nu + (e_2)_\mu (e_2)_\nu ) - (l + a 
\cos \theta)^2 g_{\mu \nu}, \quad  \quad K:= \sigma_{\mu\nu} \dot{\gamma}^{\mu}
\dot{\gamma}^{\nu}.
\end{equation}
This tensor can be obtained from the general expression of the conformal Killing-Yano tensor for the Pleba\ifmmode \acute{n}\else \'{n}\fi{}ski-Demia\ifmmode \acute{n}\else \'{n}\fi{}ski family of solutions, as presented in \cite{PhysRevD.76.084036}. Carter's constant corresponds somewhat loosely to the total angular momentum of the test body.\\
The constants of motion can be used to write the geodesic equation as a 
system of first order ODEs, e.g. \cite[p. 242]{MR1647491}:
\begin{subequations}\label{eq:eom}
\begin{align}
 \dot t &= \frac{\chi(L_z-E \chi)}{\Sigma \sin^2\theta}+\frac{(\Sigma+ a\chi)((\Sigma+a \chi)E-a L_z)}{\Sigma\Delta},\\
\dot \phi &=\frac{L_z-E \chi}{\Sigma \sin^2\theta}+\frac{a((\Sigma+a \chi)E-a L_z)}{\Sigma\Delta} ,\\
\Sigma ^2 \dot r ^2  &= R (r,E,L_z,K) := ((\Sigma+a \chi)E-aL_z)^2-\Delta K,
\label{eq:radial}\\
\Sigma ^2 \dot \theta ^2 &= \Theta (\theta, E, L_z, Q) :=  K -  \frac{(\chi E-L_z)^2}{\sin ^2 \theta}.
\label{eq:theta}
\end{align}
\end{subequations}

Note that $\Sigma + a \chi$ depends only on $r$, so the 
radial and the $\theta$ equations are separated. Moreover, they
are homogeneous in $E$ and thus for $E\neq0$ we have:
\begin{align}
R(r,E,L_z,Q)&=E ^2 R(r,1,L_E,K_E),\\
\Theta (\theta, E, L_z, Q)& =E ^2 \Theta(r,1, L_E,K_E),
\end{align} 
where $L_E=L_z/E$ and $K_E=K/E^2$. These constants 
are invariant
under affine reparametrization of the geodesic.

\subsubsection{The Trapped Set}\label{sec:trapped}
Trapping for null geodesics in black hole spacetimes describes the phenomenon that there exist null geodesics which never leave a spatially compact region of the exterior region. Consequently, these null geodesics never cross the future
or the past event horizons, $\HH^{\pm}$, and they also neither go to future
or past null infinity
$\scri^{\pm}$. From the properties of the function $R(r,E,L_z,K)$, it follows
(see \cite{grenzebach_photon_2014} and \cite{smoothness}) that 
the trapped null geodesics are those which stay at a fix value of $r$
and hence satisfy $\dot r=\ddot r=0$, which corresponds to:
\begin{equation}\label{eq:trappcondi}
R(r,L_E,K_E)=\frac{d}{dr} R(r,L_E,K_E)=0.
\end{equation}
These equations can be solved 
for the constants of motion in terms of the constant value
$r= r_{trapp}$ as 
\cite{grenzebach_photon_2014}: 
\begin{equation}\label{eq:trapparam}
K_E=\left.\frac{16 r^2 \Delta}{(\Delta')^2}\right|_{r=r_{trapp}}, \qquad aL_E=\left.(\Sigma +a \chi)-\frac{4r \Delta}{\Delta'}\right|_{r=r_{trapp}},
\end{equation}
where $\Delta'$ denotes the derivative of $\Delta$ with respect to $r$. 

The allowed values of $r_{trapp}$ are obtained from the condition 
$\Theta\geq 0$ and read  \cite{grenzebach_photon_2014}:
\begin{equation}\label{eq:areaoftrapp}
(4r\Delta-\Sigma \Delta')^2\leq 16 a^2r^2\Delta\sin^2\theta. 
\end{equation}
\section{Trapping as a Set of Directions}\label{sec:sphere}
For the following, instead of the set of trapped null geodesics we will be investigating future trapped or past trapped null geodesics. This kind of trapping exists at every point in the exterior region.
\subsection{Framework}
First we have to introduce the basic framework and notations. Let $\mathcal{M}$ be a smooth  manifold with Lorentzian metric $g$. At any point $p$ in $\mathcal{M}$ choose an orthonormal basis $(e_0,e_1,e_2,e_3)$ for the tangent space, with $e_0$ time-like and future directed. It is sufficient to treat only past directed null geodesics, as the future directed ones can be obtained by a sign flip in the parametrization. The tangent vector to any past pointing null geodesic 
at $p$ can be written as:
\begin{equation}\label{eq:tangentplus}
\dot\gamma( k|_p)|_p =\alpha ( -e_0+ k_1 e_1+ k_2 e_2+ k_3 e_3),
\end{equation} 
where $\alpha= g(\dot\gamma,e_0) >0$ and $k =(k_1,k_2,k_3)$ satisfies $|k|^2 =1$, hence $k\in \mathbb{S}^2$. The geodesic  is independent of the scaling of the tangent vector as this corresponds to an affine reparametrization for the null geodesic, so the specific value of $\alpha$ is irrelevant.
The $\mathbb{S}^2$ is often referred to as the celestial sphere of a time-like observer at $p$, whose tangent vector is given by $e_0$, e.g. \cite[p.8]{penrose_spinors_1987}.\\
For the further discussion we fix the tetrad. We can make the following definition:
\begin{definition}Let $\gamma (k|_p)$ denote a null geodesic through $p$ for which the tangent vector at $p$ is given by equation \eqref{eq:tangentplus}.
\end{definition}
\noindent It is clear that $\gamma(k_a|_p)$ and $\gamma (k_b|_p)$ are equivalent up to parametrization if $k_a=k_b$. Suppose now that $\mathcal{M}$ is the exterior region of a black hole spacetime with a complete $\mathcal{I}^\pm$ and boundary $\HH^+\cup\HH^-$. As in \cite{smoothness} we can then define the following sets on $\mathbb{S}^2$ at every point $p$:
\begin{definition}\label{def:futinf}
The future infalling set: $\Omega_{\mathcal{H}^+}(p):=  \{k\in \mathbb{S}^2 | \gamma(k|_p)\cap\mathcal{H}^+ \neq \emptyset \} $.\\
The future escaping set: $\Omega_{\mathcal{I}^+}(p):=  \{k\in \mathbb{S}^2| \gamma(k|_p)\cap\mathcal{I}^+ \neq \emptyset \}$ .\\
The future trapped set: $\mathbb{T}_+(p) := \{k\in \mathbb{S}^2 | \gamma(k|_p)\cap(\mathcal{H}^+\cup \mathcal{I}^+)  = \emptyset \}$.\\
The past infalling set: $\Omega_{\mathcal{H}^-}(p):=  \{k\in \mathbb{S}^2 | \gamma(k|_p)\cap\mathcal{H}^- \neq \emptyset \}$.\\
The past escaping set: $\Omega_{\mathcal{I}^-}(p):= \{k\in \mathbb{S}^2 | \gamma(k|_p)\cap\mathcal{I}^- \neq \emptyset \}$.\\
The past trapped set: $\mathbb{T}_- (p):= \{k\in \mathbb{S}^2 | \gamma(k|_p)\cap(\mathcal{H}^-\cup \mathcal{I}^-)  = \emptyset \}$
\end{definition}
\begin{definition}
We refer to the set $\Omega_{\mathcal{H}^-}(p)\cup \mathbb{T}_-(p)$ as the shadow of the black hole. 
\end{definition}
Note that light from a background source, i.e. not in between the black hole and the observer and sufficiently far away, can only reach the observer in the set $\Omega_{\mathcal{I}^-}(p)$ and hence the shadow will be black. For any practical purposes one can only extract information about the boundary of the shadow from an observation. In \cite{smoothness} 
it was shown that for the 
Kerr-Newman-Taub-NUT black hole the boundary of the shadow is given by the set $\mathbb{T}_-(p)$ and that this set consists of 
those directions that asymptote to the trapped null geodesics in the past.

\subsection{Degeneration for observers located on an axis of symmetry}\label{sec:ssshadow} 
The following discussions applies for observers located on an axis of
symmetry, i.e. an observer located at any regular point $p$ in the exterior region of a black hole spacetime, for which there exists a one parameter family of diffeomorphisms with closed orbits that leave $p$ invariant. This includes in particular any point in the exterior region of a spherically symmetric black hole spacetime, as well as observers located on the rotation axis of e.g. Kerr. The discussion for points of symmetry, is here treated in a separate section to introduce several important concepts needed for our main theorem. Points of symmetry are special with respect to degeneracies as it was shown in \cite{smoothness} that the shadow for observers at regular points of symmetry in the exterior of Kerr-Newman-Taub-NUT black holes is circular. \\ 
It is well-known (see e.g. \cite[p.14]{penrose_spinors_1987}) that a change of observer (i.e. an ortochronus Lorentz transformation of the tetrad)  corresponds to 
a conformal transformation on the celestial sphere, and vice versa.
Restricting oneself to
orientation preserving transformations, they are isomorphic to M\"obius transformations. A fundamental property of conformal transformations on $\mathbb{S}^2$
is that 
they map circles
into circles. As a consequence if $p_1$ and $p_2$ are points in
(possibly different) spacetimes in the family under consideration,
and both lie  on an axis of symmetry
then, upon identification of the two celestial spheres by a respective choice of
time oriented orthonormal basis, there exists a Lorentz transformation\footnote{By this we always mean an ortochronus Lorentz transformation.} $(\mathbf{LT})$ of the observer such that  $\mathbb{T}_-(p_1)= \mathbf{LT}[\mathbb{T}_-(p_2)]$. This concept is central to our argument. 
\begin{definition}\label{def:shadowdeg}
The shadows at two points $p_1$, $p_2$ are called degenerate if, upon identification of the two celestial spheres by the orthonormal basis, there exists an element of the conformal group on $\mathbb{S}^2$ that transforms $\mathbb{T}_-(p_1)$ into $\mathbb{T}_-(p_2)$.
\end{definition} 

\begin{remark}
The shadow at two points $p_1$, $p_2$ being degenerate implies that for every observer at $p_1$ there exists an observer at $p_2$ for which the shadow on $\mathbb{S}^2$ is identical. Because this notion compares structures on $\mathbb{S}^2$, the two points need not be in the same manifold for their shadows to be degenerate. Just from the shadow alone an observer can not distinguish between these two configurations.
\end{remark}

Combining the discussion above 
with the shape of the shadow at points off the axis when $a \neq 0$ described in the next subsection, we concluide that the only reliable information that 
an observer knowing to live in the exterior of a
Kerr-Newman-Taub-NUT black hole can extract from observing a circular shadow is that he/she is on an axis of the symmetry of the black hole.

\subsection{Parametrization of the Shadow for generic observers}\label{sec:shadowparam}

The shadow at any point in the exterior region of a Kerr-Newman-Taub-NUT spacetime has been explicitly obtained in 
\cite{grenzebach_photon_2014}  and its smoothness properties discussed in 
\cite{smoothness}. Here we only summarize the results that we need later.

From here on for the rest of the paper we will always assume that $a\neq0$ as $a=0$ has been treated in the previous section. Fixing the 
orthonormal tetrad \eqref{eq:tetrad} to which we will refer as ``standard observer'', the celestial sphere  can be coordinated by  standard
spherical coordinates 
$\rho \in[0,\pi]$ and $\psi \in[0,2\pi)$ so that 
\eqref{eq:tangentplus} can be written as:
\begin{equation}\label{eq:tangentsph}
\dot\gamma( \rho,\psi)|_p =\alpha ( -e_0+  e_1 \sin\rho \cos\psi+ e_2\sin\rho \sin\psi +  e_3 \cos\rho).
\end{equation}
At any point $p$ in the exterior region of a Kerr-Newman-Taub-NUT black hole
the curve $\mathbb{T}_{-}(p)$ that defines the shadow is given by the parametric
expression:

\begin{subequations}\label{eq:parametrized}
\begin{align}
\sin\psi&= \frac{\Delta'(x)\{x^2 +(l+a \cos [\theta(p)])^2\}-4 x \Delta(x)}{4ax\sqrt{\Delta(x)} \sin[\theta(p)]}\label{eq:psitox}\\ \nonumber&:=f(x,\theta,M,a,Q,l),\\
\sin\rho&= \frac{4x\sqrt{\Delta(r(p))\Delta(x)}}{\Delta'(x)(r(p)^2-x^2) + 4x \Delta(x)}\label{eq:rhotox} \\ \nonumber &:= h(x,r,M,a,Q,l),
\end{align}
\end{subequations}
where the parameter $x$ takes values in the compact
interval $[r_{min}(\theta(p)),r_{max}(\theta(p))]$ 
and  $r_{min}(\theta)$ and $r_{max}(\theta)$ are obtained by solving 
the equality case in \eqref{eq:areaoftrapp}. Geometrically they
correspond to the smallest and largest values that $r$ 
can take at the 
intersection of a cone of constant $\theta$ with the
area of trapping. The parameter $x$ 
corresponds to the asymptotic value of $r$ along
the past null
geodesic with initial tangent vector along the direction defined by
$\{ \rho(x), \psi(x)\}$. Note that the shadow curve is independent of the Manko-Ruiz parameter $C$.

In \cite{grenzebach_photon_2014} the parametrization  of the shadow curve
was in fact obtained
for the more general case of 
the Kerr-Newman-Taub-NUT-(anti-)de Sitter spacetime family. It was further
extended in  
\cite{grenzebach_photon_2015} 
to the full Pleba\ifmmode \acute{n}\else \'{n}\fi{}ski-Demia\ifmmode \acute{n}\else \'{n}\fi{}ski class. A rigorous proof of the fact that the sets $\mathbb{T}_- (p)$ and $\mathbb{T}_+ (p)$ given by this parametrization are smooth curves on the celestial sphere in the exterior region of any observer in a Kerr-Newman-Taub-NUT spacetime was given in \cite{smoothness}.

One important observation, already made in \cite{grenzebach_photon_2014},
is that the shadow for the standard observer is symmetric on the celestial sphere with respect to the $k_1=0$ plane (i.e. the great circle in the celestial
sphere defined by the meridians $\psi=\pi/2$ and $\psi = - \pi/2$) 
This is simply due to the form of equation \eqref{eq:theta} that gives two solution $\pm\sqrt{\Theta(\theta,L_E,K_E)/ \Sigma}$ for any combination of conserved quantities $L_E$ and $K_E$. Therefore if $(k_1, k_2,k_3)\in \mathbb{T}_- (p)$ 
then we always have that $(-k_1,k_2,k_3)\in \mathbb{T}_- (p)$. 
Further note that from the radial equation \eqref{eq:radial} we get immediately that if $k=(k_1,k_2,k_3)\in \mathbb{T}_+(p)$ then  $k=(k_1,k_2,-k_3)\in \mathbb{T}_-(p)$. Hence the properties of the past and the future trapped sets are equivalent. In particular this implies that if there exists a conformal map $\Psi$ from  $\mathbb{T}_-(p_1)$ to $\mathbb{T}_-(p_2)$ then there exists another conformal map, related to $\Psi$ by conjugation with a reflection about the $k_3=0$ plane, that maps $\mathbb{T}_+(p_1)$ to  $\mathbb{T}_+(p_2)$. 

An observer can only see the past, hence $\mathbb{T}_-(p)$, so we will 
concentrate on this curve in the search of degeneracies.
However the fact that the  properties  $\mathbb{T}_- (p)$ and $\mathbb{T}_+ (p)$
are equivalent 
tells us that our results
will hold true for both.

\section{Which degeneracies exist?}\label{sec:degeneracies}
The question one would like to answer is, which observers can be fully distinguished 
based on the shape of the shadow they observe. A quick inspection of the equations \eqref{eq:parametrized} shows that the shadow is invariant up to a reparametrization $x\rightarrow x/M$ as long as the following quantities are constant:
\begin{equation}\label{eq:obvious}
\theta,\ \frac{r}{M}, \ \frac{a}{M}, \ \frac{Q}{M}, \  \frac{l}{M}.
\end{equation}
With this degeneracy we see that the shape of the shadow can only be affected by the change of dimensionless parameters. There is a discrete degeneracy for two observers with: 
\begin{equation}
M_1=M_2 \quad l_1=-l_2 \quad r_1=r_2 \quad a_1=a_2 \quad Q_1=Q_2 \quad \theta_1=\pi -\theta_2
\end{equation}
In the case $l=0$ this corresponds to a reflection of the observers position with respect to the equatorial plane, while when $l\neq 0$ the spacetime itself
changes. In either case, two observers related by this transformation are fully
indistinguishable from the observation of the shadow. \\
The comparison for the shadows of two arbitrary observers is a difficult
problem and it is unclear to the authors how to determine all
possible degeneracies. We will therefore restrict ourselves to continuous degeneracies. Hence a family of observers who form a $C^1$ curve in the space of parameters for which the shadows are indistiguishable. \\
In the following we will introduce a method to systematically search for continuous degeneracies and to prove when such degeneracies do not exist. 
We will heavily rely on the fact that we have an explicit parametrization $c(x;r,\theta,M,a,Q,l)$ for the curve defining the boundary of the shadow at a point $p$ with coordinates
$r,\theta$ in the exterior region of a Kerr-Newman-Taub-NUT spacetime with parameters ($M$, $a$, $Q$, $l$) (and curve parameter $x$).\\
To studying continuous degenerations we impose that the first variation of the curve is zero. From here on we will look at the shadow as a curve in the complex plane which is obtained from the parametrization \eqref{eq:parametrized} by stereographic projection of the celestial sphere \cite[p.10]{penrose_spinors_1987}:
\begin{subequations}
\begin{align}
c(x)&= \frac{X(x)+ i Y(x)}{1-Z(x)}\label{eq:curvdef},\\
X(x)&= \sin(\rho) \sin(\psi)=h(x) \cdot f(x),\\
Y(x)&= \sin(\rho) \cos(\psi)=\pm h(x) \cdot \sqrt{1-f^2(x)}\label{eq:yparam},\\
Z(x)&= \cos(\rho)=- sgn \left(\frac{\partial h}{\partial x}\right)
\sqrt{1-h^2(x)}\label{eq:sparam}.
\end{align}
\end{subequations} 
The freedom of sign choice in \eqref{eq:yparam} comes from that fact that upon stereographic projection the symmetry with respect to the $k_1=0$ plane on the celestial sphere becomes a reflection symmetry with respect to the real axis for the curve in the complex plane. The sign in $Z$ makes the curve $C^1$
and is the right choice to describe $\mathbb{T}_{-}$ \cite{smoothness}.
If we were to describe $\mathbb{T}_{+}$ instead, the global minus sign
in front would have to be dropped. Outside the area of trapping $sgn\left(\frac{\partial h}{\partial x}\right)$ has a fixed sign. Inside the area of trapping it changes sign when $x=r$, where $h=1$ and thus $Z=0$.  
\\
From the definition of degeneracies for black hole shadows it follows that any degeneracy is characterized by a change in parameters together with a M\"obius transformation on the shadow (as the M\"obius transformation on the complex plane are equivalent to the orientation preserving conformal transformations on the Riemann sphere). Therefore when searching for continuous degeneracies we have to take the M\"obius transformation into consideration. The limitation of our result to continuous degeneracies arises from the fact that we analyze small perturbations, hence we linearize the problem.\\
The first order of the action of any member of the conformal group on $\mathbb{S}^2$ on a curve is given by:
\begin{equation}\label{eq:mobiusinftes}
    \Psi_{\epsilon}(c)= c(x) + \epsilon \vec \xi|_{c(x)} + \mathcal{O}(\epsilon^2),
\end{equation}
where $\epsilon$ is a small parameter and where $\xi$ is a conformal Killing vector field on $\mathbb{S}^2$. 
The first variation of the curve with respect to a parameter $p$ is given by: 
\begin{equation}
   c(x;p + dp) = c(x,p) +\vec V_p dp  + \mathcal{O}(dp^2), 
\end{equation}
where $dp$ is an infinitesimal change of the parameter and $V_p$ is given by $ \partial_{p} c(x,p)$. 
The most generic variation vector for a curve is then:
\begin{equation}
    \vec V = \sum_{p \in\mathcal{P}=\{r,\theta,M,a,Q,l\}}\vec V_p dp + \sum_{\xi\in Lie(Mb)}\vec \xi |_{c(x)} \epsilon_\xi.
    \end{equation} 
We can now formulate a necessary and sufficient condition for the curve to be invariant under a continuous deformation. This is the case if there exists a nontrivial combination of $dp$ and $\epsilon_\xi$ such that $V$ is tangential to the curve.
Letting $n$ be the normal to the curve $c(x,p)$, the condition is that:
\begin{equation}\label{eq:condition}
 \vec V \cdot \vec n \equiv 0
\end{equation} 

has a nontrivial solution in terms of $dp$ and $\epsilon_\xi$. Here we did not yet restrict the vector field $\xi$, however as we will discuss next, there
are a priori restrictions on the most general conformal Killing vector capable
of compensating the deformations induced by the change in parameters.

\subsection{Vector Fields from M\"obius Transformations}
By the definition of degeneracies for every observer at point $p_1$ there exists an observer at point $p_2$ who observes the exact same shadow. For our purpose we can reformulate this statement the following way: If the shadows at two points are degenerate, then there exists a M\"obius transformation that maps the stereographic projection of the shadow of a standard observer at point $p_1$ to the stereographic projection of the shadow of the shadow of the standard observer at point $p_2$.\\
As we have observed in section \ref{sec:shadowparam} the stereographic projection of the shadow of any standard observer is reflection symmetric with respect to the real line. A rather involved argument (which may be of independent
interest) is needed to show that only those conformal transformation that preserve the reflection symmetry can be used to ``counter'' the deformation from the change in parameters (as those correspond to a change between standard observers). The detailed  proof is given  in Appendix \ref{app:A}.\\
On finds that the most general such conformal Killing vector is an arbitrary
linear combination of the  three linearly independent vector fields  given by:
\begin{equation}\label{eq:availablevf}
    \vec\xi_1 = \partial_x,\qquad \vec\xi_2 = x \partial_y + y \partial_x, \qquad \vec\xi_3 = (x^2 - y^2) \partial_x + 2 xy \partial_y,
\end{equation}
in terms of Cartesian coordinates $\{ x,y\}$  on the complex plane, i.e. $z = x + i y$. 

\subsection{Conditions for Continuous Degenerations}
We now start with the explicit calculations. Most of them ae by no means difficult, but they are lengthy and have thus been performed mostly in Mathematica. Here we will describe the essential steps involved. From here on we will restrict to a domain of $x$ such that $sgn\left(\frac{\partial h}{\partial x}\right)$ does not change.  
This does not restrict our argument, as our aim is to prove that a certain quantity is zero independent of $x$. So it is equivalent to consider the problem in an open and dense interval. With:
\begin{equation}
 \vec V_p= \begin{pmatrix}
    \frac{ d(Re(c(x,p)))}{dp}\\
    \frac{ d(Im(c(x,p)))}{dp}
 \end{pmatrix},
\end{equation}
and with the normal vector to a curve parametrized by $x$ in two dimensions given by:
\begin{equation}
 \vec n=\pm \begin{pmatrix}
    \frac{ d(Im(c(x)))}{dx}\\
    -\frac{ d(Re(c(x)))}{dx}
 \end{pmatrix}, 
\end{equation}
we can calculate the various terms that show up in \eqref{eq:condition}. Note here that the sign choice in the definition of the normal vector corresponds to the choice between the inward and the outward pointing normal to the curve. Because we want to find curves with $V\cdot n=0$ it doesn't matter which orientation or normalization we choose for $n$ as long as we choose it consistently, hence we pick the plus. From equation \eqref{eq:curvdef} we directly get:
\begin{align}
    Re(c(x))&=  \frac{X(f(x),h(x))}{1-Z(h(x))},\\
    Im(c(x))&= \frac{  Y(f(x),h(x))}{1-Z(h(x))}.
\end{align}
Plugging everything in, we obtain the following result in terms of $f(x)$ and $h(x)$:
\begin{align}
\vec V_{p}\cdot \vec n=& \frac{h(x)\left(\frac{\partial f(x,p)}{\partial x}\frac{\partial h(x,p)}{\partial p}-\frac{\partial f(x,p)}{\partial p}\frac{\partial h(x,p)}{\partial x}\right)}{\sqrt{1-f^2(x)}\sqrt{1-h^2(x)}(1-\sqrt{1-h^2(x)})^2},\\
\vec\xi_1 \cdot \vec n=& \frac{\sqrt{1-h^2(x)}f(x)h(x)\frac{\partial f(x)}{\partial x}+(1-f^2(x))\frac{\partial h(x)}{\partial x}}{\sqrt{1-f^2(x)}\sqrt{1-h^2(x)}(1-\sqrt{1-h^2(x)})^2},\\
\vec\xi_2 \cdot \vec n=& \frac{h^2(x)\frac{\partial f(x)}{\partial x}}{\sqrt{1-f^2(x)}(1-\sqrt{1-h^2(x)})^2},\\
\vec\xi_3 \cdot \vec n=& \frac{h^2(x)(1-\sqrt{1-h^2(x)})\left(f(x)h(x)\frac{\partial f(x)}{\partial x} +  \frac{\partial h(x)}{\partial x}-f^2(x)\frac{\partial h(x)}{\partial x}\right) }{\sqrt{1-f^2(x)}\sqrt{1-h^2(x)}(1-\sqrt{1-h^2(x)})^4}\\
\nonumber & -\frac{f(x)h^4(x)\frac{\partial f(x)}{\partial x}}{\sqrt{1-f^2(x)}\sqrt{1-h^2(x)}(1-\sqrt{1-h^2(x)})^4}.
\end{align}
At this point it is important to note that $f(x,\theta,M,a,Q,l)$, $h(x,r,M,a,Q,l)$ and all their partial derivatives are rational functions in $x$ after multiplication with $\sqrt{\Delta(x)}$. For a list of all partial derivatives of $f(x,\theta,M,a,Q,l)$ and $h(x,r,M,a,Q,l)$ see Appendix \ref{app:B}. Hence any product of $f$, $h$ and their derivatives which contain an even number of factors
is a rational function in $x$, while any such product with an odd number of 
factors  is a rational function in $x$ after multiplication with $\Delta_r(x)$. (i.e. $h(x)f(x)\frac{\partial f(x,p)}{\partial x}\frac{\partial h(x,p)}{\partial p}$ and $ h^3(x)\left(\frac{\partial f(x,p)}{\partial x}\right)^2 \sqrt{\Delta(x)} $ are both rational functions in $x$).\\
Further we notice that away from the real axis we have $f^2(x)<1$ and outside the area of trapping we always have $h^2(x)<1$.\\ 
\begin{definition}
An degeneration is called intrinsic when there is no need to act with a M\"obius transformation to counter the deformation in the shadow due to the change in parameters.
\end{definition}
The condition for an intrinsic degeneracy of the shadow is then the existence
of a non-trivial value of $dp$ such that the following linear combination 
vanishes:
\begin{equation}\label{eq:cond0}
 \sum_{p\in \mathcal{P}} \left(\frac{\partial f(x,p)}{\partial x}\frac{\partial h(x,p)}{\partial p}-\frac{\partial f(x,p)}{\partial p}\frac{\partial h(x,p)}{\partial x}\right)dp \equiv 0,
\end{equation}
where $\mathcal P$ is the set of parameters within which we are searching for degeneracies of the shadow. 
If we now write down the general linear combination that we required to be zero in condition \eqref{eq:condition}:
\begin{equation}
 \beta \vec\xi_1 \cdot \vec n+ \alpha \vec\xi_2 \cdot \vec n+\gamma \vec\xi_3 \cdot \vec n+ \sum_{p\in \mathcal P} \vec V_{p}\cdot \vec n dp\equiv 0,
\end{equation}
we get one set of terms which are products of $f$, $h$ and their derivatives with an odd number of total powers and another set of terms with an odd number of total powers but an additional factor of $\sqrt{1-h^2(x)}$. Now if $\sqrt{1-h^2(x)}$ is not a rational function (showing this will be part of our program),
then for the above condition to be true, both sets of terms have to be equal to zero on their own, as adding a rational and an irrational function can never be equal to zero unless both functions themselves are equal to zero on their own. This gives us a system of two equations that we can solve for $\beta$ and $\gamma$:
\begin{align}\label{eq:bgexpr}
\beta=& \frac{\sum_{p\in \mathcal P} h(x)\left(\frac{\partial f(x,p)}{\partial x}\frac{\partial h(x,p)}{\partial p}-\frac{\partial f(x,p)}{\partial p}\frac{\partial h(x,p)}{\partial x}\right)dp}{2\left((1-h^2)f(x)h(x)\frac{\partial f(x)}{\partial x}-(1-f^2(x))\frac{\partial h(x)}{\partial x} \right)}\\ \nonumber &+ \alpha\frac{h^2(x)\frac{\partial f(x)}{\partial x}}{2\left( f(x)h(x)\frac{\partial f(x)}{\partial x}-(1-f^2(x))\frac{\partial h(x)}{\partial x} \right)},\\
\gamma=& \frac{\sum_{p\in \mathcal P} h(x)\left(\frac{\partial f(x,p)}{\partial x}\frac{\partial h(x,p)}{\partial p}-\frac{\partial f(x,p)}{\partial p}\frac{\partial h(x,p)}{\partial x}\right)dp}{2\left((1-h^2)f(x)h(x)\frac{\partial f(x)}{\partial x}-(1-f^2(x))\frac{\partial h(x)}{\partial x} \right)}\\ \nonumber &- \alpha\frac{h^2(x)\frac{\partial f(x)}{\partial x}}{2\left( f(x)h(x)\frac{\partial f(x)}{\partial x}-(1-f^2(x))\frac{\partial h(x)}{\partial x} \right)}.
\end{align}
Now we know that $\beta$ and $\gamma$ both are constants. With the above result this is only possible if both terms are independent of $x$ individually. 
Addding them and noticing that $\alpha =0$ is always a possibility, we conclude
that the condition for the existence of a degeneracy of the shadow within a certain set of parameters $\mathcal P$ is the existence of a non-trivial $dp$
satisfying:
\begin{equation}\label{eq:cond2}
  \frac{\partial}{\partial x}\left(\frac{\sum_{p\in \mathcal P} h(x)\left(\frac{\partial f(x,p)}{\partial x}\frac{\partial h(x,p)}{\partial p}-\frac{\partial f(x,p)}{\partial p}\frac{\partial h(x,p)}{\partial x}\right)dp}{2\left((1-h^2)f(x)h(x)\frac{\partial f(x)}{\partial x}-(1-f^2(x))\frac{\partial h(x)}{\partial x} \right)}\right)\equiv 0
\end{equation}
has a non-trivial solution in terms of the $dp$. If only the trivial solution exists than there exists no continuous degeneracy within the parameter set $\mathcal P$. We can ellaborate further on the role of $\alpha$ as follows: 
either
\begin{equation}\label{eq:cond1}
 \frac{\partial}{\partial x}\left(\frac{h^2(x)\frac{\partial f(x)}{\partial x}}{2\left( f(x)h(x)\frac{\partial f(x)}{\partial x}-(1-f^2(x))\frac{\partial h(x)}{\partial x} \right)}\right) \equiv 0,
\end{equation}
and $\alpha$ can take any value but with the the only effect of modiying both $\beta$ and $\gamma$, or \eqref{eq:cond1} does not hold,  and we must take
$\alpha=0$. In no case the validity of \eqref{eq:cond1} affects the existence
of a degeneracy. In fact, one can show that the above condition can never be satisfied but since this is of no relevance to our argument, we will omit the proof here and just assume $\alpha$ to be zero.  For the actual proof this leaves us with the following strategy:
\begin{enumerate}
\item Check whether or not intrinsic degeneracies exist using condition \eqref{eq:cond0}.
\item Check whether eventual intrinsic degeneracies can be used to eliminate parameters from the set within which one has to search for degeneracies.
\item Check that $\sqrt{1-h^2(x)}$ is an irrational function for all possible combinations of the remaining parameters in $\mathcal P$.
\item Check that the denominator of the first term in \eqref{eq:bgexpr} is not equivalent to zero for all possible combinations of parameters.
\item Check whether there exist any non-trivial solutions to \eqref{eq:cond2} for all possible combinations of the remaining parameters.
\end{enumerate}
Note that wherever in these steps we have to show that something is either equivalent to zero or not equivalent to zero the expressions we have to check are polynomials. Hence the condition is that the coefficients for every order of $x$ have to be equal to zero simultaneously, which leaves us with a system of equations that has to be satisfied. These system of equations in the steps above are of different complexities, however for most steps too involved to be solved by hand. Note that the derivation until here is independent of the detailed form of $f(x)$ and $h(x)$ and hence in principle valid for any black hole spacetime where the parametrization \eqref{eq:parametrized} exists, hence given the results in \cite{grenzebach_photon_2015} the following analysis can in principle be carried out for the entire Pleba\ifmmode \acute{n}\else \'{n}\fi{}ski-Demia\ifmmode \acute{n}\else \'{n}\fi{}ski class of black hole spacetimes.

\subsection{Continuous Degeneracies }\label{sec:main}
We now apply the above recipe to the Kerr-Newman-Taub-NUT family, which we are interested in the present work, hence our set of parameters is given by $\mathcal P =\{M,a,Q, l,r,\theta\}$ for this section. We will here re-derive the degeneracy already mentioned in \eqref{eq:obvious} to illustrate the way the method works. We start with the first point in the list, the search for intrinsic degeneracies. 
\begin{lemma}\label{lem:deg}There are two intrinsic degeneracies given by: 
\begin{equation}
  \frac{a}{M}=C_1, \qquad \frac{r}{M}=C_2, \qquad \frac{Q}{M}=C_3, \qquad \frac{l}{M}=C_4, \qquad \theta=C_5.
\end{equation}
and 
\begin{align*}
 & a\sin\theta=C_1,\qquad l+ a \cos \theta = C_2,\qquad Q+ 2a\cos \theta (l+a\cos \theta )=C_3, \\
& r=C_4, \qquad  M=C_5.
\end{align*}
\end{lemma}
\begin{proof}
The derivatives in Appendix \ref{app:B} are written such that every term in \eqref{eq:cond0}: 
\begin{equation}\begin{split}
  \left(\frac{\partial f}{\partial x}\frac{\partial h}{\partial a}\right.&\left.-\frac{\partial f}{\partial a}\frac{\partial h}{\partial x}\right)da +
  \left(\frac{\partial f}{\partial x}\frac{\partial h}{\partial M}-\frac{\partial f}{\partial M}\frac{\partial h}{\partial x}\right)dM +  \left(\frac{\partial f}{\partial x}\frac{\partial h}{\partial r}\right)dr \\
  &  +\left(\frac{\partial f}{\partial x}\frac{\partial h}{\partial Q}-\frac{\partial f}{\partial Q}\frac{\partial h}{\partial x}\right)dQ+ \left(\frac{\partial f}{\partial x}\frac{\partial h}{\partial l}-\frac{\partial f}{\partial l}\frac{\partial h}{\partial x}\right)dl-  \left(\frac{\partial f}{\partial \theta}\frac{\partial h}{\partial x}\right)d\theta\equiv 0
  \end{split}
\label{eq:first}
\end{equation}
has the same denominator. The numerator in the above equation is a polynomial in $x$ of order $11$. Hence this condition gives us a system of 11 equation. Solving this system leaves us with two degrees of freedom. One of the solution is given by $dl=ldM/M$. Insewrting this yields the following set of ODEs:
\begin{equation}
  \frac{da}{a}=\frac{dM}{M}, \qquad \frac{dr}{r}=\frac{dM}{M}, \qquad \frac{dQ}{Q}=\frac{dM}{M}, \qquad \frac{dl}{l}=\frac{dM}{M}, \qquad d\theta=0,
\end{equation}
which can be integrated to give:
\begin{equation}\label{eq:Mdeg}
  \frac{a}{M}=C_1, \qquad \frac{r}{M}=C_2, \qquad \frac{Q}{M}=C_3, \qquad \frac{l}{M}=C_4, \qquad \theta=C_5.
\end{equation}
where $C_1$, $C_2$, $C_3$, $C_4$ and $C_5$ are integration constants. 
We now explain a method that will be used several times below. 
A degeneration involving four integration constants means that (locally) the parameter
space is threaded by a congruence of curves, with any two points along the
same curve having identical
shadows. Consider now another degeneracy, independent of the previous
one. This means that the vector field tangent to the new congruence
of curves is linearly independent of the previous one.
Consider a point $p$ in parameter space where the two vectors are linearly
independent. At that point, and in fact in an open neighbourhood thereof,
the two congruences of curves are nowhere tangent to each other. It follows
that the shadow at any point in this open set is now invariant under
a two parameter family of transformations, i.e. a two-dimensional
surface in parameter space. Consider a hypersurface passing through $p$
and transverse to  the first congruence. The intersection of this hypersurface
with the invariant two-dimensional surface is necessarily a non-trivial
degeneration curve, which obviously does not belong to the first congruence. This means that we can look for linearly
independent degenerations by restricting the problem to a hypersurface
transverse to the original one. This greatly simplifies the computations.
Geometrically, the procedure is analogous to performing a gauge fixing.
In summary, the idea is to use the existing degenerations to reduce the order 
of the problem. Point (2) in the strategy outlined above refers precisely 
to this ``gauge fixing'' procedure.

Applying this strategy, the second degeneracy condition can be found without
loss of generality by setting $dM=0$ (the foliation by hypersurfaces
is given by $M = \mbox{const}$, which indeed is transverse to the congruence
of curves defined by \eqref{eq:Mdeg}). Solving the set
of equations obtained from (\ref{eq:first}) with $dM=0$ yields:

\begin{equation}
d \theta = \frac{\sin \theta}{a} d l, \quad 
da = - \cos \theta d l , \quad dQ =   2 (l + a \cos \theta) d l,
\quad dr =0, \quad dM=0,
\label{Vectorl}
\end{equation}
%
%
%
%
%
%
and can be integrated to yield: 
\begin{equation}\label{eq:ldeg}
  a\sin \theta =C_1,\ l+ a \cos \theta = C_2,\ Q+ 2a\cos \theta (l+a\cos \theta )=C_3,\ r=C_4, \ M=C_5, 
\end{equation}
where $C_1$, $C_2$, $C_3$, $C_5$ and $C_5$ are again integration constants.
\end{proof}
The first degeneracy can be ``gauge fixed'' immediately and globally
by fixing $M = \mbox{const}$ and restricting the whole
problem to this lower dimensional parameter space. We want to exploit
in a similar way the second degeneracy and reduce the problem further.
The vector field along the second degeneration can be read off directly from
\eqref{Vectorl} and it always has a non-zero component along
the $l$ direction. Thus, a suitable family of transverse hypersurfaces is $l
= \mbox{const}$.

Now we want to prove that if we set $M=\mathrm{const}$ and $l=\mathrm{const} $ then there is no further degeneracy in $\mathcal P=\{a,r,Q,\theta\}$. We start with point (3) of the recipe in the previous section. 
\begin{lemma}\label{lem:hirrat}The function $\sqrt{1-h^2(x)}$ is irrational.
\end{lemma}
\begin{proof}
We need to prove that
\begin{equation}\label{eq:irrat}
  [\Delta'(x)(r^2-x^2) + 4x \Delta(x)]^2- 16x^2 \Delta(r)\Delta(x) = P(x)^2
\end{equation}
admits no solution where $P(x)$ is a polynomial on $x$. The leading term
in the right-hand side is $4 x^6$, which combined with the fact that
a global sign in $P(x)$ is irrelevant, shows that $P(x)$ must be of the form
$P(x)= 2 x^3 + K_2 x^2 + K_1 x + K_0$. The zero, first and fifth order coefficients in \eqref{eq:irrat} are immediately solved to give:
\begin{equation*}
  K_2 = - 6 M, \quad \quad
  K_1 = - 2 \epsilon (r^2 + 2 \beta) , \quad \quad
  K_0 = 2 \epsilon M r^2,
  \end{equation*}
where $\epsilon = \pm 1$ and $\beta := a^2 - l^2 +q^2$.
The choice $\epsilon=-1$ makes
$P(x) \equiv \Delta'(x) (r^2 -x^2) + 4 x \Delta(x)$ and equation (\ref{eq:irrat})
becomes $16 x^2 \Delta(r) \Delta(x)=0$, which is impossible for $r$
in the exterior region. For the choice 
$\epsilon =1$ the coefficients in $x^4$ and $x^3$ in (\ref{eq:irrat})
impose, respectively:
\begin{align*}
  2 M r + \beta & =0, \\
  - 2 ( 2 M r + \beta ) - \Delta(r) + \beta & =0.
\end{align*}
Since in the exterior region $r> r_+ > 0$, the first requires $\beta < 0$
and the second $\Delta(r) = \beta < 0$,
which is impossible.
We conclude that $\sqrt{1-h^2(x)}$ is an irrational function.
\end{proof}

Next we check that the denominator in \eqref{eq:cond2} is non-trivial for all allowed parameter combinations.
\begin{lemma}\begin{equation}
\left((1-h^2)f(x)h(x)\frac{\partial f(x)}{\partial x}-(1-f^2(x))\frac{\partial h(x)}{\partial x} \right)\not \equiv 0
\end{equation}\end{lemma}
\begin{proof} Plugging in the parametrization \eqref{eq:parametrized} we get:
\begin{equation}\label{eq:g1}
\frac{\sqrt{\Delta(r)}\{-x(\Delta'(x))^2-2\Delta(x)(\Delta'(x)-x \Delta''(x))\} g_1(x)}{16a^2 x \sqrt{\Delta(x)}\sin^2\theta(4x \Delta(x)+(r^2-x^2)\Delta'(x))^3} \not\equiv 0,
\end{equation}
where $g_1(x)$ is given by the following polynomial of order six:
\begin{equation}
\begin{split}
  g_1(x)=&16x\left[x^2+(l+a\cos\theta)^2\right]\Delta(r)\left(2\left[x^2+(l+a\cos\theta)^2\right]\Delta'(x)-8x\Delta(x)\right) - \\
&\left(4x \Delta(x)+(r^2-x^2)\Delta'(x)\right)\left(
32 a^2 x (r^2-x^2)\sin^2\theta \right .  \\
& \left .
+ (r^2 + (l + a \cos \theta)^2 )  \left (
-32 x \Delta(x) 
+ 8 (x^2 + (l + a \cos \theta)^2)\Delta'(x) 
 \right ) 
\right). 
\end{split}
\end{equation}
The first factor in the numerator of \eqref{eq:g1} is clearly non-zero for an observer in the exterior region.  The second factor 
 is a poynomial in $x$ with leading term $-4 x^3$, hence non-indetically zero.
The zeroth order coefficient for $g_1(x)$ is: 
\begin{equation}\label{eq:zerothcoeff}
-32M^2r^2 (l+a\cos \theta )^2 (r^2+(l+a\cos \theta )^2),
\end{equation}
thus the only way this can be zero for an observer in the exterior region is if $l=-a\cos \theta$. Plugging that in for the other coefficient we get that the first order coefficient is given by $-64M Q^2 r^4$ and the fifth order coefficient is given by $-192M(-Q^2+2Mr)$. Those can never be equal to zero at the same time which finishes the proof for this lemma.
\end{proof}
When we plug the parametrization into the numerator 
inside the parenthsis in \eqref{eq:cond2} we get this is equal to:
\begin{equation}\label{eq:T6}
\frac{\{-x(\Delta'(x))^2-2\Delta(x)(\Delta'(x)-x \Delta''(x))\} g_2(x)}{2 a^2 \sqrt{\Delta(x)}\sin \theta (4x \Delta(x)+(r^2-x^2)\Delta'(x))^3}, 
\end{equation}
where $g_2(x)$ is given by the following polynomial of order five:
\begin{equation}
\begin{split}
    g_2(x)=& 2a \left[x^2+(l+a \cos \theta)^2\right](4x \Delta(x)+(r^2-x^2)\Delta'(x)) \cdot \\
    &(2Q dQ+ 2a da+ (2r-2M) dr)-\\
    &\left\{\vphantom{\frac{1}{2}}16x(x^2-r^2)(da+a \cot \theta d\theta)\Delta(x)+ 16a Q x \left[r^2+(l+a \cos \theta)^2\right]dQ+  \right.\\
    & 16 a^2x \left[r^2+(l+a \cos \theta)^2\right]da+ (2x-2M)\cdot \\
    &\left[\vphantom{\frac{1}{2}} 8a r  \left[x^2+(l+a \cos \theta)^2\right] dr+ 4 (r^2-x^2)(x^2+l^2-a^2 \cos^2 \theta) da \right.\\
  & \left.  \left. +\frac{4a(r^2-x^2)(2a l + \cos \theta(x^2+l^2+a^2+a^2\sin^2 \theta ))}{\sin \theta }d\theta\right]\right\}\Delta(r).
\end{split}
\end{equation}
We can now plug \eqref{eq:g1} and \eqref{eq:T6} into condition \eqref{eq:cond2} to obtain:
\begin{equation}\label{eq:cond3}
\frac{\partial}{\partial x}\left( \frac{8x\sin\theta g_2(x)}{\sqrt{\Delta(r)}g_1(x)} \right)\equiv 0.
\end{equation}
At this point we introduce the notion of a restricted degeneracy. 
\begin{definition}
A restricted degeneracy is one where a combination of parameters has to be zero instead of just being constant. 
\end{definition}
Since a degeneracy is defined by a curve in parameter space, two things may happen. Either the curve
is tangent to the submanifold of parameter space defined by the 
restrited degeneracy, or it is transverse to it. Im the latter case,
the curve leaves immediately the submanifold, and hence the degeneration
curve must exist away from the submanifold. It follows
that the only degeneracies that one could be missing by the general analysis
are those satisfying not only that the parameters are zero, but also that their
variation is zero, so that the curve is tangent to the restricted submanifold.

An example of a restricted degeneracy is 
$\sin \theta=0$, under which 
condition \eqref{eq:cond3} is obviously satisfied. By the argument above, 
the corresponding degeneration curves must satsify $d\theta=0$. The other parameters can vary arbitrarily in this case. Thus,
with a slight abuse, we recover the degeneracy on the rotation axis. 
Of course, the argument in this case is not fully sound since
it ignores  the fact that the coordinate system and the shadow parametrization breaks down on the axis.  This argument just serves the purpose to illustrate
the concept of a restricted degeneracy.  The situation on the rotation axis was treated properly in section \ref{sec:ssshadow}. In the following we will always assume that $\sin \theta \neq0$. 
\begin{lemma}\label{lem:onlyg2}\begin{equation}\frac{\partial}{\partial x} \left( \frac{8x\sin\theta g_2(x)}{\sqrt{\Delta(r)}g_1(x)} \right)\equiv 0 \quad
\quad \Longrightarrow \quad \quad g_2(x)\equiv 0. \end{equation}
\end{lemma}
\begin{proof} First note that \eqref{eq:cond3} can only be true if either $g_1(x)= B x g_2(x)$ for some non-zero constant $B$, or if $g_2(x)\equiv0$.
We will now exclude the first possibility. Note that the zeroth order coefficient of $x g_2(x)$ is zero and with the zeroth order coefficient for $g_1(x)$ given in \eqref{eq:zerothcoeff}. Thus, the only chance for the two to be proportional is if: 
\begin{equation}
l=-a\cos \theta. \label{lexp}
\end{equation}
We fixed $l$ this requires for 
$d \theta =\cot \theta a^{-1} da$
to hold. However plugging these two condition into  $g_2(x)$ we get that now 
its zeroth order coefficient also vanishes. Thus not only the zeroth, but also the first term in $g_1(x)$ must be zero. 
Plugging (\ref{lexp})  into $g_1(x)$ it follows that its first order coefficient is  given by $-64MQ^2 r^4$, which vanishes only if $Q=0$. Setting $Q=0$ and $dQ=0$ everywhere, the first order coefficient in $g_2(x)$ zero, while 
the  second order coefficient of $g_1(x)$ is $96M^2 r^4$. This is manifestly
non.zero and we reach a contradiction. Thus, the only possibility is $g_2(x)=0$
and the Lemma is proved.
\end{proof}

From this lemma, the remaining task is to show that
there exists no non-trivial solution for $g_2(x)\equiv0$ which is equivalent to condition \eqref{eq:cond0}. We emphasize, in particular, that the previous lemma
already implies that all degenerations of the shadow must be intrinsic.

The next Theorem, which is the main result of this paper, proves that there are no more degeneracies than those
already found.
\begin{thm}The only continuous degenerations of the black hole shadow for observers located at coordinate position $r$, $\theta$ in the exterior region of Kerr-Newman-Taub-NUT black holes with parameters $M$, $a$, $Q$ and $l$ are given for observers such that their parameters have the same value for all the following functions:  
\begin{equation}
  \frac{a}{M}=C_1, \qquad \frac{r}{M}=C_2, \qquad \frac{Q}{M}=C_3, \qquad \frac{l}{M}=C_4, \qquad \theta=C_5.
\end{equation}
or 
\begin{equation}
  a\sin \theta =C_1,\ l+ a \cos \theta = C_2,\ Q+ 2a\cos \theta (l+a
\cos \theta)=C_3,\ r=C_4\ M=C_5.
\end{equation}
\begin{proof}
The two degeneracies have already been derived in Lemma \ref{lem:deg}. Given Lemma \ref{lem:onlyg2} we know that the condition for degeneracies to exist is given by \eqref{eq:cond3}. 
The only thing remaining to show is that $g_2(x)\equiv0$ has no non-trivial solutions. The highest order coefficient is given by:
\begin{equation}\label{eq:coeff5}
aQ dQ+a (r-M)dr+(a^2-\Delta(r))da -\frac{a\cos\theta\Delta(r)}{\sin\theta}d\theta=0.
\end{equation}
We solve this for $dr$ and substitute back into $g_2$. This leads to a a third
order polynomial in $x$, i.e. $g_2 = \sum_{i=0}^3 w_i x^i$, and each coefficient
$w_i$ must vanish. The combination $M w_3 + w_2$ is very simple:
\begin{align*}
M w_3 + w_2 = - \frac{16 \Delta(r) M (r^2 + (l + a \cos \theta)^2) }{\sin\theta}
\left ( a \cos \theta d \theta + \sin \theta da \right ) =0.
\end{align*}
The first factor is nowhere zero in the exterior region, so we can solve
for $da$: 
\begin{align}
da =  - \frac{a \cos \theta}{\sin \theta} d \theta
\label{da},
\end{align}
and substitute back into $g_2(x)$, which  factorizes as:
\begin{align*}
g_2(x) = \frac{16 a \Delta(r) (r-x)}{r-M } g_3(x),
\end{align*}
where $g_3(x)$ is a quadratic polynomial in $x$. Obviously,
$g_2$ is identically zero only if $g_3 \equiv 0$. The highest order term of $g_3$ is:
\begin{align}
- r Q dQ + \frac{a}{\sin \theta} \left ( l(M-r) + a
M \cos \theta  \right ) d \theta
=0. \label{dQ}
\end{align}
At this point we need to split the
treatment in two cases depending on whether $Q\equiv 0$ or not.

For the case with $Q\not\equiv0$, we solve (\ref{dQ}) for $dQ$ and substitute
back into $g_3$ to obtain:
\begin{align*}
g_3(x) = \frac{a M (r-M) (r^2 + (l+ a \cos \theta)^2)}{r \sin \theta} \left (
l + a \cos \theta \right ) d \theta.
\end{align*}
Thus $g_3 (x) \equiv 0$ can can only happen if $l + a \cos \theta=0$. Taking its differential and inserting
$da$  from (\ref{da}) yields
$- a (\sin \theta)^{-1} d \theta =0$, hence 
$d\theta=da=dr=dQ=0$ and we have no continuous degeneration.

The remaining case is when $Q=0$ and $dQ=0$. We want to  impose $g_3(x) \equiv 
0$, so that in particular it must be that $g_3(x = M)=0$. Evaluating:
\begin{align*}
g_3(x=M)= \frac{M a^2 \cos \theta  d \theta 
\left ( r^2  + (l + a \cos \theta)^2 
\right )}{\sin \theta},
\end{align*}
which implies  $ \cos \theta d \theta =0$, and hence
$d \theta = 0$ (if $d\theta \neq 0$ it must be $\theta = \pi/2$ 
so that  $d\theta=0$ anyway). Consequently
$d\theta=da=dr=dQ=0$, which finishes the proof.

\end{proof}

\end{thm}

\section{Conclusion}
 
In the present work we showed that there exist only two continuous degeneracies for the shadow of any observer in the exterior region of a Kerr-Newman-Taub-NUT spacetime. In particular when one focuses on the physically relevant case of Kerr-Newman (hence $l=0$) the only continuous degeneracy is given by scaling of all parameters with the mass. If one assumes that, apart from the discrete spacetime isometries, no discrete degeneracies exist, then the result presented in this paper suggests that in principle an observer in the exterior region of a Kerr-Newman spacetime could extract the relative angular momentum $a/M$ of the black hole, as well as the relative charge $Q/M$, the relative distance $r/M$, and the angle of observation relative to the rotation axis of the black hole. Additionally, one could extract how fast one is moving in comparison to a standard observer at that point in the manifold. Preliminary calculations suggest that the same result should hold true for the Kerr-Newman-de Sitter case. The proof for this case is work in progress. It is interesting however, that from an observation of the shadow alone an observer can never conclude that the Taub-NUT charge must vanish. \\
Note that if one looks at the projection of the shadows of the standard observers on the complex plane and chooses a parabolic and a hyperbolic M\"obius transformation (see Appendix \ref{app:A}) for each standard observer such that all shadows intersect the real axis at $+1$ and $-1$, 
it turns out that the changes of the shape due to variations of $r/M$ and $Q/M$ are extremely small. Hence reading off these parameters from the shadow would require a very precise measurement of the shadow curve. Adding in the fact that the light sources can be rather messy, the observational task is certainly formidable, so that at least in the foreseeable future there is little hope that from the shape of the shadows alone one can extract in practice
more than a rough estimate on $a/M$. However, from a theoretical point of view it seems plausible (and our results are
a strong indication in this direction) that one can extract very
detailed information about a black hole
just by looking at it. 


\section*{Acknowledgements}

M.M. acknowledges financial support under the project
FIS2015-65140-P (MINECO, Spain and FEDER, UE).

\appendix

\section{M\"obius transformation}\label{app:A}

The Riemann sphere $\mathbb{S}^2$ can be blobally
parametrized by stereographic projection by means of $\C:=\mathbb{C} \cup
\{ \infty \}$. A M\"obius transformation is a map 
\begin{align*}
\chi: \C & \longrightarrow \C, \\
c & \longrightarrow \chi(z) := \frac{a z + b}{c z + d}, \quad \quad
a,b,c,d \in \mathbb{C}, \quad \quad ad-bc = 1.
\end{align*}
The set of all M\"obius transformations define a group, denoted by $\Moeb$.

This group is isomorphic to the set of positively oriented conformal maps of 
$\mathbb{S}^2$ endowed with the standard round metrid.

In this appendix we prove the following theorem. 
\begin{thm}
\label{InvThm}
Let $\C := \mathbb{C} \cup \{ \infty \}$ be the Riemann sphere and 
$\Moeb$ 
the set of M\"obieus transformations. Let 
$c: \mathbb{S}^1 \longrightarrow \C$ be an embedding.
If there exists $\chi \neq \Id_{\C}$ that leaves $c$ invariant as a 
set, then $c$ is a generalized circle (i.e. a circle or a straight line with the 
point at infinity attached), or there exists $n \in \mathbb{N}$ such that
$\chi^n = \Id_{\C}$ and $c$ is conjugate to a closed curve invariant
under rotations of angle $\frac{2\pi m}{n}$, $m \in \mathbb{Z}$ around the
origin of $\mathbb{C}$.
\end{thm}

By ``invariant as a set'' we mean that there is a diffeomorphism
$f : \mathbb{S}^1 \longrightarrow \mathbb{S}^1$ such that
$\chi \circ c = c \circ f$ (the image of $c$ and $\chi \circ c$
are obviously the same).
A closed embedded curve $c$ is {\it conjugate} to another
closed embedded curve $c_1$ if there exists $\chi_1 \in \Moeb$ such that
$\chi_1 \circ c = c_1$.

\vspace{3mm}

{\it Proof:}
We first note that the problem is invariant under conjugation:
for any  $\xi \in \Moeb$ the conjugate curve $c_{\xi} := 
\xi \circ c$ is invariant (as a set) under
the conjugate transformation $\chi_\xi := \xi \cdot \chi \cdot \xi^{-1}$, 
as it is obvious from:
\begin{align*}
\chi_\xi \circ c_{\xi} = 
(\xi \circ \chi \circ \xi^{-1} ) \circ (\xi \circ c ) =
\xi \circ (\chi \circ c) = \xi \circ c \circ f = c_{\xi} \circ f.
\end{align*}
It is well-known that 
all M\"obius transformations (different from the identity) can be
classified by conjugation into four 
disjoint classes:
parabolic, elliptic, hyperbolic or loxodromic. Each class admits
a canonical representative, in the sense that any element in the class
is conjugate to this representative. The representatives can be chosen as
follows:
\begin{align}
& \mbox{Parabolic:} \quad \quad & & \chi_P(z)= \frac{z}{1+z} &&  \label{list} \\
& \mbox{Elliptic:} \quad \quad & & \chi_E(z)=   e^{i \theta} z, \quad \quad
&&0 \neq \theta \in \mathbb{R}_{\mbox{mod } 2 \pi} \nonumber \\
& \mbox{Hyperbolic:} \quad \quad & & \chi_H(z)=  e^{\lambda} z \quad \quad 
&& \lambda \in 
\mathbb{R} \setminus \{ 0 \} \nonumber \\
& \mbox{Loxodromic:} \quad \quad & & \chi_L(z)=   k z \quad \quad && k \in \mathbb{C}
\setminus \{ \mathbb{R} \} \quad \quad \mbox{and} \quad \quad |k| \neq 1
\nonumber 
\end{align}
Thus, we may assume without loss of generality
that the transformation $\chi$ leaving $c$
invariant is one of these canonical
transformations. Obviously $\chi^m, m \in \mathbb{Z}$ also leaves $c$ invariant.
The action of $\chi^m$ is immediate to write down in the elliptic,
hyperbolic and loxodromic canonical cases.
In the parabolic case, a simple inductive argument
shows that:
\begin{align*}
\chi_P^m (z ) = \frac{z}{m z+1}, \quad \quad m \in \mathbb{Z}.
\end{align*} 
Thus, it follows
that  the cyclic group $\{ \chi^n ; n \in \mathbb{Z} \}$ is finite 
(i.e. $\chi^m = \mbox{Id}$ for some $m \in \mathbb{Z}$)
if and only if $\chi$ is elliptic
and $\frac{\theta}{2\pi} \in \mathbb{Q}_{\mbox{mod } 1}$.

Let us consider first the loxodromic, hyperbolic and parabolic cases. We start
by showing that  the embedded loop $c$ must pass through the
origin $z=0$  of the complex plane. Let $0 \neq z_0 \in \mathbb{C}$ be any point on the curve, i.e. $z_0 \in \mbox{Im} (c)$
and define, for each $m \in \mathbb{Z}$, $z_m := 
\chi^m (z_0) \in \C$. From invariance of the
curve under $\chi$, all points in the sequence $\{ z_m \}$ lie on
the image of the curve. From compactness of $\mbox{Im} (c) \subset \C$
it follows that the set of accumulation points of $\{ z_m\}$ is non-empty
and a subset of $\mbox{Im}(c)$.

When $\chi$ 
is hyperbolic or loxodromic, the canonical form is
$\chi_k(z) := k z$ with $|k| \neq 1$. The sequences are now 
$z_m := \chi_k^m (z_0) =  k^m z_0$. If $|k| > 1$, the 
sequence converges to $z=0$ as $m \rightarrow -\infty$. If 
$|k| <1$, the sequence converges to $z=0$ as $m \rightarrow \infty$. In either
case $z=0$ is an accumulation point, so the loop $c$ passes throught $z=0$.
When $\chi$ is parabolic, the sequence is
$z_m = \frac{z_0}{m z_0 + 1 }$ which converges to $z=0$ as $|m| \rightarrow 
\infty$, and we reach the same conclusion.

We can now show that a loxodromic M\"obius transformation does not leave
any closed embedded loop invariant. Let us take differentials in the
invariance equation $\chi \circ c = c \circ f$ and evaluate at
the invariant point $p := \{ z=0\}$:
\begin{align*}
d \chi |_p (\dot c) = \dot{f} |_{c^{-1} (p)} \dot{c},
\end{align*}
which simply states the fact that
the differential map of $\chi_p$ must 
preserve the direction of $\dot{c}|_p$ (it may change its scale, but not 
the direction). The differential of $\chi(z)= k z$ at $z=0$
is  $d \chi |_{z=0} = k$. Thus, this differential
acts on a vector $v$ by scaling with $|k|$ and rotating by $\mbox{arg} (k)$.
When $k$ is not real, all vectors $v \neq 0$ change direction and
we reach  a contradiction. Thus, no embedded loop is invariant
under a  loxodromic M\"obius transformation.

We next consider the hyperbolic case. The canonical representative
is now $\chi = \chi_{H}$. Let $\xi$ be a rotation of the form
$\xi(z) = e^{i\alpha} z$ $\alpha \in \mathbb{R}$. Upon conjugation with
$\xi$, the map $\chi_H$ remains unchanged.  The conjugate
curve 
$\xi \circ c$ passes though $z=0$, and the parameter $\alpha$ can be
adjusted so that its tangent vector there points along the
real axis $x$.
Since $c$ is an embedded curve, there is a 
neighbourhood $U$ of $z=0$ such that $U \cap c$ is connected 
and in fact a graph over
the real axis. 
After restricting $U$ if necessary we may assume that
$U$ is an open disk centered at $z=0$.
We consider the curve $c_U := c \cap U$ from now on.
This curve can be parametrized by $x$, i.e.
$c(x) = x + i y(x)$ where $y(x)$
is a smooth function of $x \in 
(-\epsilon,\epsilon)$. The parameter $\lambda$
in the definition of 
$\chi_H$ can be assumed to be negative
(if it were positive simply replace 
$\chi_H$ by $\chi_H^{-1}$). Then $\chi_H$ maps $U$ into itself, and leaves
the curve $c_U$ invariant. So, it must be the case
that, for all $x \in (- \epsilon, \epsilon)$:
\begin{align}
e^{\lambda} ( x + i y(x) ) = x' (x) + i y(x'(x))  \quad
\Longleftrightarrow  \quad y(e^{\lambda} x) = e^{\lambda} y(x).
\label{invariance}
\end{align}
where $x'(x)$ indicates the reparametrization of the curve induced
by the M\"obius transformation $\chi_H$. Define the function
$P(u) := e^{-\lambda u} y(e^{\lambda u})$. By construction, $P(u)$ is smooth
on $(-\infty,\lambda^{-1} \ln \epsilon)$. In terms of $P$, the function $y(x)$ restricted to $x>0$ takes the form $y(x)= x P(\lambda^{-1} \ln x)$. 
The  invariance property (\ref{invariance}) becomes, when
applied at the point $x = e^{\lambda u}$:
\begin{align*}
P(u+1)= e^{-\lambda u} e^{-\lambda} y(e^{\lambda u} e^{\lambda} )
= e^{-\lambda u} y (e^{\lambda u} )= P(u).
\end{align*}
So $P(u)$ is a periodic function of period one. We can now compute
the derivative of $y(x)$ (prime denotes derivative with respect to $u$):
\begin{align*}
\frac{dy(x)}{dx} = 
P(\lambda^{-1} \ln x) + \lambda^{-1} P' |_{\lambda^{-1} \ln x}.
\end{align*}
If $P(u)$ is not a constant function
the combination $P(u) + \lambda^{-1} P'(u)$ does not converge 
as $u \rightarrow - \infty$. To show this, take the sequence $u_n = u_0 - n$
with $u_0 \in [-1,0)$ defined by the condition that
$P(u_0)$ attains the supremum of $P(u)$
and another sequence  $u'_n = u_1 - n$ where
$u_1 \in [-1,0)$ is the value where $P(u)$ attains the infimum.
By periodicity, the sequences $P(u_n)$ and $P(u'_n)$ are both constant.
Moreover,  $P'$ vanishes on all points $u_n$ and $u'_n$.
Thus, the sequences $\{ P(u_n) + \lambda^{-1} P'(u_n) \}$
and $\{ P(u'_n) + \lambda^{-1} P'(u'_n) \}$
 converge to the same
limit if and only if $P(u_0) = P(u_1)$, i.e.  if the function $P(u)$
is constant, as claimed.
As a consequence, $\frac{dy}{dx}$ converges
as $x \rightarrow 0^+$ if and only if $P(u) = a$ for some constant $a$,
or equivalently iff $y(x) = ax$. Since,
in our setup, $\frac{dy}{dx}=0$ at $x=0$  we conclude that $y(x)=0$.
We have proved this fact in a neighbourhood $U$ of $0$,
but this extends to the whole loop $c$ by applying repeatedly the transformation
$\chi_H$. In summary, we have shown that the only  embedded
loops invariant under the canonical
representative $\chi_H$ of hyperbolic M\"obius transformations is
the line $(x, y=0)$, and arbitrary rotations thereof around
the origin. We now use the property that
M\"obius transformations map generalized circles 
into generalized circles, and conclude that 
an embedded loop which is not a generalized circle can never be invariant
under a hyperbolic M\"obius transformation.

We want to use a similar argument for the parabolic case. To that aim, it is 
preferable to use a different representative.
More precisely, recall that for $\chi = \chi_P$ given in 
(\ref{list}) the invariant embedded loop  $c$ necessarily passes
through $z=0$. Let us apply a conjugation
with the inversion map $\hat{\xi}(z) = -1/z$. The conjugate
$\widehat{\chi}_P = \hat{\xi} \circ \chi_P \circ \hat{\xi}^{-1}$ is given by
$\widehat{\chi}_P (z) = z -1$ and the conjugate loop
$\hat{c}:= \hat{\xi} \circ c$ passes through the point at infinity. 
Consider the vector field:
\begin{align*}
\zeta = z^2 \partial_z + \overline{z}^2 \partial_{\overline{z}}.
\end{align*}
This field is smooth in a neighbourhood of the point at infinity.
Indeed,
the vector field
$\partial_{x'} =
\partial_{z'} + \partial_{\overline{z}'}$ is clearly smooth in a neighbourhood
of zero. The inversion map $z' = - \frac{1}{z}$ transforms
this neighbourhood of zero into a neighbourhood of infinity and transforms
the vector field $\partial_{x'}$ into $\zeta$, from which smoothness follows.
In the coordinates $\{x,y\}$ defined by $z= x+ i y$ this vector
field takes the form:
\begin{align*}
\zeta = \left (x^2 - y^2 \right ) \partial_x + 2 xy \partial_y.
\end{align*}
The property of invariance of an embedded loop 
under a M\"obius transformation 
is
preserved by reparametrizations of the curve, so we are free to choose
the parametrization of $\hat{c}$. However,  we must make sure that the 
parameter is smooth everywhere, including a neighbourhood
of infinity. To that aim
we choose to parametrize $\hat{c}$ with
arc length $s$ with respect to the round sphere metric:
\begin{align}
ds^2 = \frac{1}{\left ( 1+ \frac{1}{4} (x^2 + y^2 ) \right )^2}
(dx^2 + dy^2),
\label{metric}
\end{align}
which extends smoothly to the point at infinity.
As before, let $0 \neq z_0= \hat{c}(s_0) = (x_0,y_0) \in \mathbb{C}$
be a point on the curve. 
From the condition that the  tangent vector
$T|_p$ of the curve is unit with respect to (\ref{metric}),
there exists $\alpha \in  [0,2\pi)$ such that:
\begin{align*}
T|_p = F|_p \left ( \cos \alpha \partial_x + \sin \alpha
\partial_y \right ),
\end{align*}
with $F|_p$ determined by:
\begin{align*}
F|_p = \left . 1 + \frac{1}{4} \left ( x^2 + y^2 \right ) \right |_{(x_0,y_0)}.
\end{align*}
We compute the scalar product with the vector $\zeta$ to find:
\begin{align*}
\langle T|_p, \zeta |_p \rangle
=
\left . \frac{\cos \alpha (x^2 - y^2 ) + 2 \sin \alpha xy }
{1 + \frac{1}{4} \left ( x^2 + y^2 \right )}
\right |_{(x_0,y_0)}.
\end{align*}
Consider now the sequence of points $\{ z_m = (x_0 -m, y_0)\} $. From
invariance under $\widehat{\chi}_P$, they also also lie on the curve
$\hat{c}$. In fact, the set $\mbox{Im}(\hat{c})$
defines a periodic submanifold, in the sense that 
a unit translation along the $x$ axis leaves it invariant.
As a consequence, all the the tangent 
vectors $T_{p_m}$
of the curve at each point $z_m$ must be parallel to each other
(in the natural euclidean sense of the term).
Hence $\alpha$ is the same for all $z_m$. Let us
compute the limit along the sequence of the scalar product
$\langle T|_{p_m}, \zeta |_{p_m} \rangle$:
\begin{align*}
\lim_{m \longrightarrow \infty}
\langle T|_{p_m}, \zeta |_{p_m} \rangle
=
\lim_{m \longrightarrow \infty}
\frac{
\cos \alpha ( (x_0-m)^2 - y_0^2 ) + 2 \sin \alpha (x_0 - m)
y_0 }{1 + \frac{1}{4} \left ( (x_0-m)^2 + y_0^2 \right )} =
4 \cos \alpha.
\end{align*}
Given that the curve is smooth evergywhere, including infinity,
and that the sequence $\{ z_m\}$ converges to the point at infinity,
it follows that all the tangent vectors $T|_{p_m}$ must converge,
namely to the unit tangent vector $T_{\infty} $ to the curve there. The scalar
products above must then converge to a single finite value, and this must happen
independently of the initial point $z_0$.
Since the limit depends on $\alpha$ we conclude that $\alpha$ must be
the same for all points along the curve. If $\alpha = \frac{\pi}{2}$
or $\alpha = \frac{3\pi}{2}$
then the curve would be an infinite collection of vertical lines
in the $\{ x,y\}$ plane, all of them passing though the point
at infinity and the curve $\hat{c}$ would not be embedded.
Thus the tangent vector $T_p$ must have a non-zero component
along the $x$ axis everywhere along the curve.
This implies that it can be described as a graph $y(x)$ on the $x$
axis. Since $y(x)$ must reach a local maximum and $\alpha$ vanishes there
we conclude that $\alpha=0$ at all points, and hence that 
$y=y_0 = \mbox{const}$.
So, the embedded loop $\hat{c}$ must be the straight line $y =y_0$. This claim is
for embedded curves invariant under  the parabolic transformation $z \rightarrow 
z -1$. Upon conjugation, and using again that M\"oebius transformations
map generalized circles into generalized circes, we conclude
that the only
embedded closed loops invariant under a parabolic transformation
are generalized circles.

It only remains to consider the elliptic case, i.e. $\chi = \chi_E$.
Since $\chi_E$ is a rotation of angle $\theta$
of the complex plane around its origin, the invariant embedded loop $c$
defines a figure invariant under a rotation of angle $\theta \neq 2 \pi k$,
$k \in \mathbb{Z}$.
Consider the set of all angles $\beta \in (0,2\pi)$
under which this figure is invariant
and let $\beta_0$ be its infimum. If $\beta_0 =0$, the curve
must be a circle. If $\beta_0$ is different from zero, then there must exist 
$n \in \mathbb{N}$ such that $\beta_0 = \frac{2\pi}{n}$ (if such $n$
did not exist, define $n \in \mathbb{N}$ by
$n \beta_0 <2\pi < (n+1) \beta_0 $,  the angle $(n+1) \beta_0
- 2\pi$ is positive, smaller that $\beta_0$ and belongs to the set
of rotation angles that leave the figure invariant, which is a contradiction.)
Thus $\beta_0 = \frac{2\pi}{n}$ and in fact all other symmetry angles
must be a multiple of this (by a similar argument as before).
The number $n$ is called the {\it order of symmetry} of the figure.
In summary, the closed embedded loop $c$ is invariant under $\chi_E$
if and only if it is a circle centered at zero, or a figure with a discrete
rotational symmetry of order $n$. The statement of the theorem then follows
once again from the fact that the collection of generalized circles is preserved
under M\"obius transformations. $\hfill \Box$

\vspace{5mm}

As discussed in the main text, the shadow curve for suitable chosen
observers at any point in the class of black hole spacetimes under
consideration here has the property of being reflection symetric.
In precise terms, let  the map $r: \C
\longrightarrow \C$ be defined by reflection 
with respect to the real axis  $y=0$, i.e. $r(z) = \zb$.
A closed embedded loop
$c : \mathbb{S}^1  \longrightarrow \C$
is {\bf reflection symmetric} if there exists
a smooth map $f_1 : \mathbb{S}^1 \longrightarrow \mathbb{S}^1$ such that
$r \circ c = c \circ f_1$. One checks immediately that $f_1$ is 
a diffeomorphism of $\mathbb{S}^1$ (in fact an orientation reversing
diffeomorphism).  
Our aim is to determine
which elements $\chi \in \Moeb$ have the property that
the conjugate curve $\chi \circ c$ is also 
reflection symmetric.  Thus, we want to impose the 
condition that there exists a diffemorphism $f_2 : \mathbb{S}^1
\longrightarrow \mathbb{S}^1$ such that
$r \circ \chi \circ c = \chi \circ c \circ f_2$, which in turn is equivalent to
$\chi^{-1} \circ r \circ \chi \circ r^{-1} \circ c \circ f_1 = c \circ f_2$,
i.e. to:
\begin{align*}
\chi^{-1} \circ r \circ \chi \circ r^{-1} \circ c  = c \circ f,
\end{align*}
where $f := f_2 \circ f_1^{-1}$ is an orientation preserving
diffeomorphism of $\mathbb{S}^1$.
The map 
$\widetilde{\chi}:=
\chi^{-1} \circ r \circ \chi \circ r^{-1} $ is by construction an element
of the M\"obius group, and leaves the loop defined by $c$ invariant (as 
a submanifold). From Theorem \ref{InvThm} it follows
that $\widetilde{\chi}$ is the indentity map, unless either 
$\mbox{Im}(c)$
is conjugate to a figure with discrete rotational symmetry of order $n$
and, in addition, $\widetilde{\chi}$ is conjugate
to $\chi_{m,n} := z \rightarrow e^{i \frac{2 \pi m}{n}} z$ for some integer $m$
between $-n$ and $n$,  or else  $c$ is a generalized circle.  

In this paper we are interested in M\"obius transformations sufficiently close to
the identity that map reflection symmetric curves into 
reflection symmetric curves. Since, for fixed $n$
$\{ \xi_{m,n}; - n < m < n \}$ is discrete, it is disjoint
to a sufficiently
small neighbourhood of the identity map  $\Id_{\mathbb{C}}$, and we can ignore
the case of discrete rotational symmetry of order $n$.  Also, 
we restrict ourselves to non-degenerate spacetimes points, where the shadow curve
is not a generalized circle (for simplicity we call such
curves ``non--circular''). So, we conclude that
$\widetilde{\chi}$ must be the identity map, i.e.:
\begin{align*}
\chi^{-1} \circ r \circ \chi \circ r^{-1} = \Id_{\C} \quad \quad
\Longleftrightarrow \quad \quad
r \circ \chi \circ r^{-1}  = \chi.
\end{align*}
Letting $\chi$ correspond to the $\mbox{SL}(2,\mathbb{C})$ matrix:
\begin{eqnarray*}
\left ( \begin{array}{cc}
\alpha & \beta \\
\gamma & \delta \\
\end{array}\right ),
\end{eqnarray*}
it is immediate to compute that $ r \circ \chi \circ r^{-1}$ corresponds to the
$\mbox{SL}(2,\mathbb{C})$ matrix:
\begin{eqnarray*}
\left ( \begin{array}{cc}
\overline{\alpha} & \overline{\beta} \\
\overline{\gamma} & \overline{\delta} \\
\end{array} \right ).
\end{eqnarray*}
Thus, is $\chi$ is sufficiently close to the identity map and the
reflection symmetric curve $c$  is non-circular, it must be the
case that $\chi \in SL(2,\mathbb{R})$, i.e. all 
$\alpha, \beta, \gamma, \delta$ are real
parameters.

Our second aim is to identify the infinitessimal transformations with
generate this subgroup of M\"obius transformations. Consider
a one parameter subgroup $\tau: \mathbb{R} \longrightarrow 
\mbox{SL}(2,\mathbb{C})$ of $\mbox{SL}(2,\mathbb{C})$
and denote by
$\chi_{\tau(s)}$, $s \in \mathbb{R}$ 
the corresponding curve in the M\"obius group. A straightforward
computation gives, for each $z \in \mathbb{C}$:
\begin{align*}
\frac{d \chi_{\tau(s)}(z)}{d s}  =
\beta_0 + \left ( \alpha_0 - \delta_0 \right ) - \gamma_0 z^2,
\end{align*}
where $\alpha_0 = \left . \frac{d \alpha(s)}{ds} \right |_{s=0}$,
$\beta_0 = \left . \frac{d \beta(s)}{ds} \right |_{s=0}$,
$\gamma_0 = \left . \frac{d \gamma(s)}{ds}\right |_{s=0}$,
$\delta_0 = \left . \frac{d \delta(s)}{ds}\right |_{s=0}$. The condition that
the curve $\tau(s)$ takes values in $\mbox{SL}(2,\mathbb{C})$ requires that $\delta_0
= - \alpha_0$. Thus, the infinitessimal generator of this one-parameter
subgroup  is:
\begin{align*}
\xi = 
\left ( \beta_0 + 2 \alpha_0 z - \gamma_0 z^2 \right )  \partial_z 
+ 
\left ( \overline{\beta_0} + 2 \overline{\alpha_0} \, \zb- 
\overline{\gamma_0} \, \zb^2 \right )  \partial_{\zb}.
\end{align*}
Thus if we restrict ourselves to the subgroup of transformation 
preserving the reflection symmetry of a non-circular curve $c$,
the generators are: 
\begin{align*}
\xi =  \beta_0 \left ( \partial_z + \partial_{\zb} \right )
+2 \alpha_0 \left ( z\partial_z + \zb \partial_{\zb} \right )
- \gamma_0 \left ( z^2\partial_z + \zb^2 \partial_{\zb} \right ),
\quad \quad 
\alpha_0,\beta_0,\gamma_0 \in \mathbb{R}.
\end{align*}
In terms of Cartesian coordinates $\{ x,y\}$  on the stereograhic plane,
i.e. $z = x + i y$, this vector field becomes:
\begin{align*}
\xi =  \beta_0 \partial_x 
+2 \alpha_0 \left ( x \partial_x + y \partial_y \right )
- \gamma_0 \left ( \left ( x^2 - y^2 \right ) \partial_x
+ 2 x y \partial_{y} \right ).
\end{align*}
So, the three generators of  M\"obius transformations preserving 
reflection symmetry turn out to be  the translations along the
$x$ axis $\xi_1 = \partial_x$, the dilations about the origin
$\xi_2 = x \partial_y + y \partial_x$ and a third conformal Killing vector
given by
$\xi_3 = (x^2 - y^2) \partial_x + 2 xy \partial_y$. These vector fields
generate a Lie algebra with structure constants:
\begin{align*}
[ \xi_1, \xi_2] = \xi_1, \quad \quad
[ \xi_1, \xi_3] = 2\xi_2, \quad \quad
[\xi_2,\xi_3] = \xi_3.
\end{align*}
Note that the subset of reflection symmetric
transformations that leave the origin $\{x=0,y=0\}$
invariant is generated by $\{\xi_2,\xi_3\}$, which is, naturally,
a two-dimensional subalgebra. Another observation is that 
the only element in $\{ \xi_1,\xi_2,\xi_3\}$ which is a Killing vector
of $\mathbb{C} \cup \{ \infty \}$ endowed with the spherical metric
$ds^2 = \left ( 1 + \frac{1}{4} (x^2 + y^2) \right )^{-2} (dx^2 + dx^2)$,
is $4 \xi_1 + \xi_3$ (and its constant multiples). This Killing field
corresponds to rotations of the sphere leaving invariant the 
antipodal points for which the corresponding equator maps onto the real axis by stereographic projection.

\section{Partial derivatives of $f$ and $h$}\label{app:B}
\begin{subequations}
\begin{align}
\frac{\partial f}{\partial a}&= \frac{ 4x \Delta^2 -a^2\{x^2 +(l+a \cos  \theta )^2\}) \Delta' -\Delta(4a^2 x  + (x^2+l^2-a^2 \cos^2 \theta )\Delta')}{4 a^2 x \Delta^{3/2}\sin \theta }\\
\frac{\partial f}{\partial l}&= \frac{ l \{x^2 +(l+a \cos  \theta )^2\} \Delta' -2\Delta(2l x  + (l+a \cos \theta )\Delta')}{4 a x \Delta^{3/2}\sin \theta }\\
\frac{\partial f}{\partial M}&= -\frac{\{x^2 +(l+a \cos  \theta )^2\}(2\Delta + x \Delta') +4 x^2 \Delta}{4 a x \Delta^{3/2}\sin \theta }\\
\frac{\partial f}{\partial Q}&= -\frac{Q(\Delta'\{x^2 +(l+a \cos  \theta )^2\}+4 x \Delta)}{4ax\Delta^{3/2} \sin \theta }\\
\frac{\partial f}{\partial \theta}&=-\frac{2a (l+a\cos \theta ) \Delta'\sin^2 \theta + \cos \theta (\Delta'\{x^2 +(l+a \cos  \theta )^2\}-4 x \Delta)}{4 a x \Delta ^{1/2}\sin^2 \theta }\\ \label{eq:fxder}
\frac{\partial f}{\partial x}&=\frac{\{x^2 + (l+a \cos  \theta )^2\}((M-x)^3-M(M^2-a^2-Q^2+l^2))}{2 a x^2 \Delta^{3/2}\sin \theta }
\end{align}
\end{subequations} 

\begin{subequations}
\begin{align}
\frac{\partial h}{\partial a}&=\frac{4a x(-8 x \Delta(r)\Delta(x)+ (\Delta(r)+\Delta(x))((r^2-x^2)\Delta'(x)+4x\Delta(x)) }{\sqrt{\Delta(x)\Delta(r)}((r^2-x^2)\Delta'(x)+4x\Delta(x))^2}\\
\frac{\partial h}{\partial Q}&=\frac{4Q x(-8 x \Delta(r)\Delta(x)+ (\Delta(r)+\Delta(x))((r^2-x^2)\Delta'(x)+4x\Delta(x)) }{\sqrt{\Delta(x)\Delta(r)}((r^2-x^2)\Delta'(x)+4x\Delta(x))^2}\\
\frac{\partial h}{\partial l}&=\frac{4l x(8 x \Delta(r)\Delta(x)- (\Delta(r)+\Delta(x))((r^2-x^2)\Delta'(x)+4x\Delta(x)) }{\sqrt{\Delta(x)\Delta(r)}((r^2-x^2)\Delta'(x)+4x\Delta(x))^2}\\
\frac{\partial h}{\partial M}&=\frac{4x\left(r\Delta(x)(-4x\Delta(x)+(x^2-r^2)\Delta'(x)) + \Delta(r)(2(r^2+x^2)\Delta(x)+x(x^2-r^2)\Delta'(x)\right) }{\sqrt{\Delta(x)\Delta(r)}((r^2-x^2)\Delta'(x)+4x\Delta(x))^2}\\
\frac{\partial h}{\partial r}&=\frac{2 x\Delta(x)(4x\Delta(x)\Delta'(r)+(\Delta'(r)(r^2-x^2)-4r\Delta(r))\Delta'(x)) }{\sqrt{\Delta(x)\Delta(r)}((r^2-x^2)\Delta'(x)+4x\Delta(x))^2}\\\label{eq:hxder}
\frac{\partial h}{\partial x}&= \frac{2(r^2-x^2)\Delta(r)((x-M)^3+ M(M^2-a^2-Q^2 +l^2))}{\sqrt{\Delta(x)\Delta(r)}((r^2-x^2)\frac{\Delta'(x)}{2}+2x\Delta(x))^2} .
\end{align}
\end{subequations} 
\newcommand{\arxivref}[1]{\href{http://www.arxiv.org/abs/#1}{{arXiv.org:#1}}}
\newcommand{\mnras}{Monthly Notices of the Royal Astronomical Society}
\newcommand{\prd}{Phys. Rev. D}
\newcommand{\apj}{Astrophysical J.}

\bibliographystyle{amsplain}
\bibliography{shadow,kerr}

\providecommand{\bysame}{\leavevmode\hbox to3em{\hrulefill}\thinspace}
\providecommand{\MR}{\relax\ifhmode\unskip\space\fi MR }
\providecommand{\MRhref}[2]{%
  \href{http://www.ams.org/mathscinet-getitem?mr=#1}{#2}
}
\providecommand{\href}[2]{#2}
\begin{thebibliography}{10}

\bibitem{bardeen_black_1973}
J.~M. Bardeen, \emph{Black {Holes} 215-39}, CRC Press, January 1973 (en).

\bibitem{1973blho.conf..215B}
J.~M. {Bardeen}, \emph{{Timelike and null geodesics in the {Kerr} metric.}},
  Black Holes (Les Astres Occlus) (C.~{Dewitt} and B.~S. {Dewitt}, eds.), 1973,
  pp.~215--239.

\bibitem{MR1647491}
S.~Chandrasekhar, \emph{The mathematical theory of black holes}, Oxford Classic
  Texts in the Physical Sciences, The Clarendon Press, Oxford University Press,
  New York, 1998, Reprint of the 1992 edition. \MR{1647491}

\bibitem{scalarhair}
P.~V.~P. Cunha, C.~A.~R. Herdeiro, E.~Radu, and H.~F. R\ifmmode~\acute{u}\else
  \'{u}\fi{}narsson, \emph{Shadows of {K}err black holes with and without
  scalar hair}, International Journal of Modern Physics D \textbf{25} (2016),
  no.~09, 1641021.

\bibitem{doeleman_event-horizon-scale_2008}
S.~S. Doeleman, J.~Weintroub, A.~E.~E. Rogers, R.~Plambeck, R.~Freund, R.~P.~J.
  Tilanus, P.~Friberg, L.~M. Ziurys, J.~M. Moran, B.~Corey, K.~H. Young, D.~L.
  Smythe, M.~Titus, D.~P. Marrone, R.~J. Cappallo, D.~C.-J. Bock, G.~C. Bower,
  R.~Chamberlin, G.~R. Davis, T.~P. Krichbaum, J.~Lamb, H.~Maness, A.~E. Niell,
  A.~Roy, P.~Strittmatter, D.~Werthimer, A.~R. Whitney, and D.~Woody,
  \emph{Event-horizon-scale structure in the supermassive black hole candidate
  at the {Galactic} {Centre}}, Nature \textbf{455} (2008), 78--80.

\bibitem{grenzebach_aberrational_2015}
A.~Grenzebach, \emph{Aberrational effects for shadows of black holes},
  arXiv:1502.02861 [gr-qc] (2015), arXiv: 1502.02861.

\bibitem{grenzebach_photon_2014}
A.~Grenzebach, V.~Perlick, and C.~L\"ammerzahl, \emph{Photon regions and
  shadows of {Kerr}-{Newman}-{NUT} black holes with a cosmological constant},
  Physical Review D \textbf{89} (2014), no.~12, 124004.

\bibitem{grenzebach_photon_2015}
A.~Grenzebach, V.~Perlick, and C.~L{\"a}mmerzahl, \emph{Photon regions and
  shadows of accelerated black holes}, International Journal of Modern Physics
  D \textbf{24} (2015), no.~09, 1542024, arXiv: 1503.03036.

\bibitem{griffiths2009}
J.~B. Griffiths and J.~Podolsk{\`y}, \emph{Exact space-times in {Einstein's}
  general relativity}, Cambridge University Press, 2009.

\bibitem{hioki_measurement_2009}
K.~Hioki and K.~Maeda, \emph{Measurement of the {Kerr} spin parameter by
  observation of a compact object's shadow}, Physical Review D \textbf{80}
  (2009), no.~2, 024042.

\bibitem{PhysRevD.76.084036}
D.~Kubiz\ifmmode~\check{n}\else \v{n}\fi{}\'ak and
  P.~Krtou\ifmmode~\check{s}\else \v{s}\fi{}, \emph{Conformal {Killing}-{Yano}
  tensors for the {Pleba\ifmmode \acute{n}\else \'{n}\fi{}ski}-{Demia\ifmmode
  \acute{n}\else \'{n}\fi{}ski} family of solutions}, Phys. Rev. D \textbf{76}
  (2007), 084036.

\bibitem{li_measuring_2014}
Z.~Li and C.~Bambi, \emph{Measuring the {Kerr} spin parameter of regular black
  holes from their shadow}, Journal of Cosmology and Astroparticle Physics
  \textbf{2014} (2014), no.~01, 041--041, arXiv: 1309.1606.

\bibitem{Manko-Ruiz}
V.~S. Manko and E.~Ruiz, \emph{Physical interpretation of the {NUT} family of
  solutions}, Classical and Quantum Gravity \textbf{22} (2005), no.~17, 3555.

\bibitem{TaubNUTsingularities}
J.~G. Miller, \emph{Global analysis of the {Kerr}-{Taub}-{NUT} metric}, Journal
  of Mathematical Physics \textbf{14} (1973), no.~4, 486--494.

\bibitem{NUT1963}
E.~Newman, L.~Tamburino, and T.~Unti, \emph{{Empty}-{space} generalization of
  the {Schwarzschild} metric}, Journal of Mathematical Physics \textbf{4}
  (1963), no.~7, 915--923.

\bibitem{smoothness}
C.~F. {Paganini} and M.~A. {Oancea}, \emph{{Smoothness of the future and past
  trapped sets in {Kerr}-{Newman}-{Taub}-{NUT} spacetimes}}, In preparation
  (2017).

\bibitem{penrose_spinors_1987}
R.~Penrose and W.~Rindler, \emph{Spinors and {Space}-{Time}: {Volume} 1,
  {Two}-{Spinor} {Calculus} and {Relativistic} {Fields}}, Cambridge University
  Press, Cambridge u.a., May 1987 (English).

\bibitem{SCHEE2009}
J.~Schee and Z.~Stuchlik, \emph{Optical phenomena in the field of braneworld
  {Kerr} black holes}, International Journal of Modern Physics D \textbf{18}
  (2009), no.~06, 983--1024.

\bibitem{synge_escape_1966}
J.~L. Synge, \emph{The escape of photons from gravitationally intense stars},
  Monthly Notices of the Royal Astronomical Society \textbf{131} (1966), 463.

\bibitem{Takahashi2010}
R.~Takahashi and M.~Takahashi, \emph{Anisotropic radiation field and trapped
  photons around the {Kerr} black hole}, A\&A \textbf{513} (2010), A77.

\bibitem{Teo2003}
E.~Teo, \emph{Spherical photon orbits around a {Kerr} black hole}, General
  Relativity and Gravitation \textbf{35} (2003), no.~11, 1909--1926.

\end{thebibliography}

\end{document}